\documentclass[11pt]{article}

\usepackage[bookmarks]{hyperref}
\usepackage{amssymb}
\usepackage{amsmath}
\usepackage{amsthm}
\usepackage{stmaryrd}
\usepackage[affil-it]{authblk}
\usepackage{graphicx,color,colordvi}
\usepackage{bbm}
\usepackage{cleveref}
\usepackage{pgfplots}
\DeclareMathOperator*{\argmin}{arg\,min}

\def\openone{\leavevmode\hbox{\small1\kern-3.8pt\normalsize1}}

\def\CC{\mathbb{C}}
\def\RR{\mathbb{R}}

\def\NN{\mathbb{N}}

\def\PP{\mathbb{P}}
\def\EE{\mathbb{E}}

\newtheorem{theorem}{Theorem}
\newtheorem{lemma}{Lemma}
\newtheorem{proposition}{Proposition}
\newtheorem{corollary}{Corollary}

\theoremstyle{definition}

\newtheorem{example}{Example}

\newtheorem{condition}{Condition}

\def\reff#1{(\ref{#1})}
\def\eps{\varepsilon}

\newcommand{\supp}{\mathop{\rm supp}\nolimits}

\newcommand{\tr}{\mathop{\rm Tr}\nolimits}

\newcommand{\im}{\mathop{\rm Im}\nolimits}
\newcommand{\spec}{{\rm sp}}

\newcommand{\cA}{{\cal A}}
\newcommand{\cB}{{\cal B}}
\newcommand{\cD}{{\cal D}}
\newcommand{\cE}{{\cal E}}
\newcommand{\cF}{{\cal F}}
\newcommand{\cG}{{\cal G}}

\newcommand{\cH}{{\cal H}}

\newcommand{\cP}{\mathcal{P}}

\newcommand{\cX}{{\cal X}}

\def\d{\mathrm{d}}
\def\e{\mathrm{e}}

\usepackage{graphicx}
\usepackage{setspace}
\usepackage{verbatim}
\usepackage{subfig}

\usepackage{float}

\theoremstyle{definition}

\theoremstyle{remark}
\newtheorem{remark}{Remark}

\numberwithin{equation}{section}

\newcommand{\indc}{{\mathbf{1}}}
\DeclareRobustCommand\openone{\leavevmode\hbox{\small1\normalsize\kern-.33em1}}

\newcommand{\id}{\rm{id}}
\newcommand{\be}{\begin{equation}}
\newcommand{\ee}{\end{equation}}
\newcommand{\bea}{\begin{eqnarray}}
\newcommand{\eea}{\end{eqnarray}}
\newcommand{\beas}{\begin{eqnarray*}}
	\newcommand{\eeas}{\end{eqnarray*}}

\begin{document}
	\bibliographystyle{abbrv}
	
	\title{Finite blocklength and moderate deviation analysis \\of hypothesis testing of correlated quantum states\\ and application to classical-quantum channels with memory
		}
	\author[a]{Cambyse Rouz\'e}
           \author[a]{Nilanjana Datta}
	\affil[a]{\small Statistical Laboratory, Centre for Mathematical Sciences, University of Cambridge, Cambridge~CB30WB, UK}
	\maketitle

	\begin{abstract}
	Martingale concentration inequalities constitute a powerful mathematical tool in the analysis of problems in a wide variety of fields ranging from probability and statistics to information theory and machine learning. Here we apply techniques borrowed from this field to quantum hypothesis testing, which is the problem of discriminating quantum states belonging to two different sequences $\{\rho_n\}_{n}$  and $\{\sigma_n\}_n$. We obtain upper bounds on the finite blocklength type II Stein- and Hoeffding errors, which, for i.i.d. states, are in general tighter than the corresponding bounds obtained by Audenaert, Mosonyi and Verstraete in \cite{AMV12}. We also derive finite blocklength bounds and moderate deviation results for pairs of sequences of correlated states satisfying a (non-homogeneous) factorization property. Examples of such sequences include Gibbs states of spin chains with translation-invariant finite range interaction, as well as finitely correlated quantum states. We apply our results to find bounds on the capacity of a certain class of classical-quantum channels with memory, which satisfy a so-called \textit{channel factorization property} - both in the finite blocklength and moderate deviation regimes. 
	\end{abstract}
	\section{Introduction}\label{sec_intro}
	
	\subsection*{Quantum Hypothesis Testing}
The goal of binary quantum hypothesis testing is to determine the state of a quantum system, given the knowledge that it is one of two specific states ($\rho$ or $\sigma$, say), by making suitable measurements on the state. In the language of hypothesis testing, one considers two hypotheses -- the {\em{null hypothesis}} $H_0 : \rho$ and the {\em{alternative hypothesis}} $H_1 :\sigma$.
The measurement done to determine the state is given most generally by a POVM $\{T, \mathbb{I} - T\}$ where $0\le T\le \mathbb{I}$, and $\mathbb{I}$ denotes the identity operator acting on the Hilbert space of the quantum system. Adopting the nomenclature from classical hypothesis testing, we refer to $T$ as a {\em{test}}. 
There are two associated error probabilities: 
\begin{align*} {\hbox{Type I error:}} \quad \alpha(T) &:= \tr\left(( \mathbb{I} -  T)\rho\right), \quad {\hbox{and}} \quad {\hbox{Type II error:}} \quad \beta(T) := \tr\left(T \sigma\right), 
	\end{align*}
which are, respectively, the probabilities of erroneously inferring the state to be $\sigma$ when it is actually $\rho$ and vice versa.
There is a trade-off between the two error probabilities, and 
	there are various ways to optimize them, depending on whether or not 
	the two types of errors are treated on an equal footing. In the case of {\em{symmetric hypothesis testing}}, one minimizes the total probability of error $\alpha(T)+\beta(T)$, whereas in {\em{asymmetric hypothesis testing}} one minimizes the type II error under a suitable constraint on the type I error.

	Quantum hypothesis testing was originally studied in the {\em{asymptotic i.i.d.~setting}}, in which, instead of a single copy, multiple (say $n$) 
	identical copies of the state were assumed to be available, and a joint measurement on all of them was allowed. 
The optimal error probabilities were
evaluated in the asymptotic setting ($n \to \infty$) and shown to decay exponentially in $n$. The decay rates were quantified by different statistical distance measures
in the different cases: in symmetric hypothesis testing it is given by the so-called quantum Chernoff distance \cite{Aetal07, NussbaumSzkola}; in asymmetric hypothesis
testing, the optimal decay rate of the type II error probability, when evaluated under the constraint that the type I error is less than a given threshold value,
is given by the quantum relative entropy \cite{HP91, ON00}, whereas, when evaluated under the constraint that the type I error decays with a given exponential speed, it is given by the so-called
Hoeffding distance \cite{HM7, N06, OH04}. The type II errors in these two cases of asymmetric hypothesis testing, are often referred to as the Stein error and the Hoeffding error, respectively.
   
The consideration of the asymptotic i.i.d.~setting in quantum hypothesis testing is, however, of little practical relevance, since in a realistic scenario only finitely many copies ($n$) of a state are available. More generally, one can even consider the hypothesis testing problem involving a {\em{finite}} sequence of states $\{\omega_n\}_n$, where for each $n$, $\omega_n$ is one of two states $\rho_n$ or $\sigma_n$, which need not be of the i.i.d.~form: $\rho_n= \rho^{\otimes n}$ and $\sigma_n=\sigma^{\otimes n}$. 
We refer to the hypothesis testing problem in these non-asymptotic scenarios as {\em{finite blocklength quantum hypothesis testing}}, the name ``finite blocklength'' referring to the finite value of $n$. Finding bounds on the error probabilities in these scenarios is an important problem in quantum statistics and
quantum information theory. To our knowledge, this problem has been studied thus far only by Audenaert, Mosonyi and Verstraete \cite{AMV12}. They obtained bounds for both the symmetric and asymmetric cases mentioned above, in the non-asymptotic but i.i.d.~scenario.

	In this paper we focus on asymmetric, finite blocklength quantum hypothesis testing, and find improved upper bounds on the Stein- and Hoeffding errors, in comparison to those obtained
in \cite{AMV12} in the i.i.d. setting. We also find upper bounds on the same quantities in the case of non i.i.d. states satisfying a factorization property. Our framework can also be applied to the analysis of the case where the error of type I converges sub-exponentially with a rate given by means of a moderate sequence, extending the results recently found in \cite{HC17,CTT17}. Finally, we apply our results to the problem of finding bounds on capacities of a certain type of classical-quantum channels with memory, both in the finite blocklength case and in the asymptotic framework of moderate deviations.\\\\
In the case of uncorrelated states, we obtain our results by use of {\em{martingale concentration inequalities}}. Concentration inequalities deal with deviations of functions
of independent random variables from their expectation, and provide upper bounds on tail probabilities of the type $\PP(|X-\EE[X]|\ge t)$ which are exponential in $t$; here $X$ denotes a random variable which is a function of independent random variables. These simple and yet powerful inequalities have turned out to be very useful in the analysis of various problems in different branches of mathematics, such as pure and applied probability theory (random matrices, Markov processes, random graphs, percolation), information theory, statistics, convex geometry, functional analysis and machine learning. Concentration inequalities have been established using a host of different methods. These include martingale methods, information-theoretic methods, the so-called ``entropy method'' based on logarithmic Sobolev inequalities, the decoupling method, Talagrand's induction method etc. (see e.g. \cite{RS13,BLM13} and references therein). In this paper, we apply two inequalities, namely the Azuma-Hoeffding inequality \cite{H63,A67} and the Kearns-Saul inequality \cite{KS98} (which have been established using martingale methods and hence fall in the class of so-called {\em{martingale concentration inequalities}}) to quantum hypothesis testing in the i.i.d.~setting. Moreover, the proofs of the results we obtain in the case of quantum hypothesis testing for correlated states are reminiscent of this framework. We include a brief review of martingales and these inequalities in \Cref{preliminaries}. 
To our knowledge, martingale concentration inequalities have had rather limited applications in quantum information theory thus far (see e.g.~\cite{EW16,IK10}). We hope that our use of these inequalities in finite blocklength and moderate deviation analyses of quantum hypothesis testing will lead to further applications of them in studying quantum information theoretic problems.

	\subsection*{Quantum Stein's lemma and its refinements}

	Consider the quantum hypothesis testing problem in which the state $\omega_n$ which is received is either $\rho_n $ or $\sigma_n$, the latter being 
states on a finite-dimensional Hilbert space $\cH_n$. The type I and type II errors for a given test $T_n$ (where $0 \le T_n \le \mathbb{I}_n$ and $\mathbb{I}_n$ is the 
identity operator on $\cH_n$), are given by
\begin{align}\label{error-probs}
\alpha(T_n)= \tr[(\mathbb{I}_n - T_n)\rho_n]\text{ and }\beta(T_n) = \tr[T_n \sigma_n].
\end{align}
As mentioned in the Introduction, in the asymmetric setting, one usually optimizes the type II error $\beta(T_n)$ under one of the following constraints 
on the type I error $\alpha(T_n) $: 
$(i)$  $\alpha(T_n) $ is less than or equal to a fixed threshold value $\eps \in (0,1)$ or $(ii)$  $\alpha(T_n) $ satisfies an exponential constraint  $\alpha(T_n) \le e^{-nr}$, for some fixed parameter $r >0$. The optimal type II errors are then given by the following expressions, respectively:
\begin{align}
\beta_n(\eps)&:=\inf_{0 \leq T_n \leq \mathbb{I}_n}\{\beta(T_n)|\alpha(T_n)\le \eps\}\label{qSteinII}\\
\tilde{ \beta}_n(r)&:=\inf_{0 \leq T_n \leq \mathbb{I}_n}\{\beta(T_n)|\alpha(T_n)\le \e^{-nr}\}\label{qHoeffII}.
\end{align} 
We refer to $\beta_n(\eps)$ as the \textit{type II error of the Stein type} (or simply the {\em{Stein error}}), and we refer to $\tilde{ \beta}_n(r)$ as the \textit{type II error of the Hoeffding type} (or simply the {\em{Hoeffding error}}).

In the i.i.d.~setting,  $\rho_n:=\rho^{\otimes n} $  and $\sigma_n:=\sigma^{\otimes n}$, with $\rho$ and $\sigma$ being states on a finite-dimensional Hilbert space $\cH$, and $\cH_n \simeq \cH^{\otimes n}$.  Explicit expressions
of the type II errors defined in \Cref{qSteinII} and \Cref{qHoeffII} are not known even in this simple setting. However, their behaviour in the asymptotic limit ($n \to \infty$) is known. 
The asymptotic behaviour of $\beta_n(\eps)$ is given by the well-known quantum Stein lemma \cite{HP91,ON00}: 
	\begin{align*} \lim_{n \to \infty}- \frac{1}{n} \log \beta_n(\eps)& =D(\rho||\sigma) \quad \forall \, \eps \in (0,1),
\end{align*} 
	where $D(\rho||\sigma)$ denotes the quantum relative entropy defined in \Cref{qrelent}.

The asymptotic behaviour of $\tilde{ \beta}_n(r)$ is given in terms of the so-called Hoeffding distance: For any $r>0$,
	\begin{align*} \lim_{n \to \infty}- \frac{1}{n} \log \tilde{\beta}_n(r)& = H_r(\rho||\sigma):= - \inf_{0\le t <1} \left\{ \frac{ t r + \log \left(\tr (\rho^t \sigma^{1-t}) \right)}{1-t}\right\}.
\end{align*} 
The problem of finding error exponents can be mapped (in the i.i.d.~case) to the problem of characterizing the probability that a sum of $n$ i.i.d.~random variables makes an order-$n$ deviation from its mean, which is the subject of \textit{large deviations} and Cram\'{e}r's theorem. In fact, it is known that in the context of Stein's lemma, allowing the error of type II to decay exponentially, with a rate smaller than Stein's exponent, $D(\rho\|\sigma)$, the error of type I decays exponentially, with a rate given by
\begin{align*}
	\sup_{0<\alpha<1}\frac{\alpha -1}{\alpha}[r-D_\alpha(\rho\|\sigma)],
\end{align*}
where $D_{\alpha}(\rho\|\sigma)$ is the so-called $\alpha$-R\'{e}nyi divergence:
\begin{align*}
	D_\alpha(\rho\|\sigma)=\frac{1}{\alpha-1}\log\operatorname{Tr}(\rho^\alpha\sigma^{1-\alpha}).
\end{align*}
If instead the error of type II is restricted to decay exponentially with a rate greater than Stein's exponent, the error of type I converge exponentially to $1$, with a rate given by
\begin{align*}
	\sup_{1<\alpha}\frac{\alpha-1}{\alpha}[r-D^*_\alpha(\rho\|\sigma)],
\end{align*}
where $D_\alpha^*(\rho\|\sigma)$ is the so-called Sandwiched $\alpha$-R\'{e}nyi divergence:
\begin{align*}
	D_\alpha^*(\rho\|\sigma)	 =\frac{1}{\alpha-1}\log\operatorname{Tr}(\rho^{1/2}\sigma^{\frac{1-\alpha}{\alpha}}\rho^{1/2})^\alpha.
\end{align*}
These phenomena are the manifestation of a coarse-grained analysis.\\\\
A more refined analysis of the {\em{type~II error exponent}}, $(-\log \beta_n(\eps))$, is given by its {\em{second order asymptotic expansion}}, which was derived independently by Li \cite{L14}, and Tomamichel and Hayashi \cite{TH13}. It can be expressed as follows:
	\begin{align}\label{soa}
	\frac{1}{n}\log\beta_n(\eps)=-\,D(\rho||\sigma)+\frac{s_1(\eps)}{\sqrt{n}}+\mathcal{O}\left(\frac{\log n}{n}\right),
	\end{align}
	where the second-order coefficient $s_1(\eps)$ displays a Gaussian behaviour given by
	\begin{align}\label{s1s1}
		s_1(\eps):=-\Phi^{-1}(\eps)\sqrt{V(\rho\|\sigma)}.
		\end{align}
	Here $\Phi$ denotes the cumulative distribution function (c.d.f.)~of a standard normal distribution, and $V(\rho||\sigma)$ is called the {\em{quantum information variance}} and is defined in \Cref{V}. Both first order and second order asymptotics of the type~II error exponent have been generalized to contexts beyond the i.i.d.~setting under different conditions on the states $\rho_n$ and $\sigma_n$ (see e.g.~\cite{DPR16} and references therein).  The problem of finding second order asymptotic expansions can actually be mapped (in the i.i.d.~case) to the one of characterizing the probability that a sum of i.i.d.~random variables makes an order-$\sqrt{n}$ deviation from its mean, which is the subject of \textit{small deviations} and the Central Limit- and Berry Esseen Theorems. \\\\
%	In practice, the error of type II is here allowed to decay sub-exponentially with a sub-exponential term given as a function of $\eps$.\\\\
	Quantum Stein's lemma and second order asymptotics both deal with the convergence of the type II error when the type I error is assumed to be smaller than a pre-fixed constant threshold value $\eps$. However, as mentioned above, imposing the error of type II to decay exponentially with a rate smaller than Stein's rate implies that the error of type I itself decays exponentially. In this paper we carry out a `hybrid analysis' in which we allow the error of type I to decay sub-exponentially with $n$, the error exponent taking the form $\eps_n:=\exp(n a_n^2)$, with $\{a_n\}_{n\in\NN}$ being a so-called moderate sequence\footnote{Such a sequence has the property $a_n\to0$, but $\sqrt{n} a_n\to\infty$, as $n\to \infty$.}. As shown in \cite{HC17,CTT17}, this problem can be mapped (in the i.i.d.~case) into the problem of characterizing the probability that a sum of i.i.d.~random variables makes an order-$a_n$ deviation from its mean. This is the subject ot \textit{moderate deviations}, hence justifying the name `hybrid analysis'. Note that even the classical counterpart of this analysis was done relatively recently, see e.g.~\cite{A14,PV10,S12}.
	\\\\
	Large, moderate and small deviations belong to the asymptotic setting. On the other hand, relatively little is known, about the behaviour of the type II errors ${ \beta}_n(\eps)$ and $\tilde{ \beta}_n(r)$ (for some $r>0$) in the case of finite blocklength, i.e.~for a fixed, finite value of $n$. As mentioned earlier, Audenaert, Mosonyi and Verstraete \cite{AMV12} considered the i.i.d.~case and derived bounds on the quantities $\beta_n(\eps)$ and $\tilde{\beta}_n(r)$ in the asymmetric setting, as well as bounds on the corresponding quantity in the symmetric setting. For example, their bounds on $\beta_n(\eps)$ (see Theorem 3.3 and Equation (35) of \cite{AMV12}) can be expressed as follows:
\begin{align}\label{boundss}
	-f(\eps) \le Q(n,\eps)	\le g(\eps)
	\end{align}	
	where 
\begin{align}\label{qne}
Q(n,\eps) := \frac{1}{\sqrt{n}}\left(\log{\beta_{n}(\eps)} + nD(\rho\|\sigma)\right),
\end{align}
and 
	\begin{align}\label{fgg}
		f(\eps)= 4\sqrt{2}\log\eta \log(1-\eps)^{-1},\qquad g(\eps)=4\sqrt{2} \log \eta \log\eps^{-1},
		\end{align}
		with $\eta:=1+\e^{1/2D_{3/2}(\rho\|\sigma)}+\e^{-1/2D_{1/2}(\rho\|\sigma)}$.

	\subsection*{Our contribution}
	\subsubsection*{Quantum hypothesis testing for uncorrelated and correlated states}
	In this paper, we obtain upper bounds on the optimal type II errors (namely, the Stein and Hoeffding errors) for finite blocklength quantum hypothesis testing, in the case in 
which the received state $\omega_n$ is one of two states $\rho_n$ and $\sigma_n$, where $\rho_n$ and $\sigma_n$ are each given by tensor products of $n$ (not necessarily identical)
states, and hence also for i.i.d.~states. We also derive similar bounds when $\rho_n$ and $\sigma_n$ satisfy the following \textit{upper-factorization} property:
	\[\rho_n\le  R ~\rho_{n-1}\otimes \rho_1,~~~~\sigma_n\le S~\sigma_{n-1}\otimes \sigma_1, ~~~~~\text{ for some }~~ ~~R,S\ge 1.
	\]
This is for example the case of Gibbs states of spin chains with translation-invariant finite-range interactions,
or finitely correlated states (see \cite{FNW92,HMO07}). This class of states was studied in \cite{HMO08b} in the asymptotic framework of Stein's lemma (see also \cite{MO15}). We also consider the case of states satisfying a so-called \textit{lower-factorization property}:
	\[\rho_n\ge  R^{-1} ~\rho_{n-1}\otimes \rho_1,~~~~\sigma_n\ge S^{-1}~\sigma_{n-1}\otimes \sigma_1, ~~~~~\text{ for some }~~ ~~R,S> 1.
	\]
	Gibbs states mentioned above, i.i.d. states and certain classes of finitely correlated states, satisfy both these factorization properties.\\\\
	In the i.i.d.~case, the upper bounds that we derive for the finite blocklength regime are tighter than the ones derived in \cite{AMV12}, for all values of the parameter $\eps$ up to a threshold value (which depends on $\rho$ and $\sigma$). We also extend the recent results of \cite{HC17,CTT17}, in the moderate deviation regime, to the case of such correlated states.%	For example, we obtain the following bound on the Stein error for a given $\rho$ and $\sigma$:
\subsubsection*{{Application to classical-quantum channels}}
Quantum hypothesis testing is one of the fundamental building blocks of quantum information theory since it underlies various other informetion-theoretic tasks.
An important example of such a task is the transmission of classical information through a quantum channel. In particular, it is well-known that the analysis of information transmission through a so-called {\em{classical-quantum (c-q) channel}}\footnote{This amounts to the transmission of classical information through a quantum channel, under the restriction of the encodings being product states.} can be reduced to a hypothesis testing problem. Hence our above results on quantum hypothesis testing can be applied to find bounds on the optimal rates of transmission of information through c-q channels, both in the finite blocklength- and the moderate deviations regime. Most notably, our results on hypothesis testing of correlated quantum states (satisfying the factorization properties mentioned above) allow us to analyze the problem of information transmission through a class of c-q channels with memory. The latter are channels whose output states satisfy a non-homogeneous factorization property (see \Cref{channelcoding} below for details). We say that such channels satisfy a {\em{channel factorization property}}.
\subsection*{Layout of the paper} In \Cref{preliminaries}, we introduce the necessary notations and definitions, including the two key tools that we use, namely, relative modular operators and martingale concentration inequalities. The finite blocklength analysis of hypothesis testing for uncorrelated quantum states is done in \Cref{sec_main} (see \Cref{uncorcase,theorem6}). The bounds that we obtain are compared with previously known finite blocklength- \cite{AMV12} and second order asymptotic \cite{L14,TH13} bounds (see \Cref{fig1}). Our finite blocklength results on correlated states, introduced in \Cref{finitecorrelated}, are given by \Cref{ahi} and \Cref{cor2} of \Cref{finitecorrelated}. Moderate deviation analysis of such states is done in \Cref{moderatedev} (see \Cref{moderate} and \Cref{moderatenh}). Our results are applied to classical-quantum channels with memory in \Cref{channelcoding} (see \Cref{prop,eq29}).

	\section{Notations and Definitions}\label{preliminaries}
\subsection*{Operators, states and relative modular operators}

	Given a finite-dimensional Hilbert space $\cH$, let $\cB(\cH)$ denote the algebra of linear operators acting on $\cH$ and $\cB_{sa}(\cH) \subset \cB(\cH)$ denote the set of self-adjoint operators. Let $\cP(\cH)$ be the set of positive semi-definite operators on $\cH$ and $\cP_{+}(\cH) \subset  \cP(\cH)$ the set of (strictly) positive operators. Further, let $\cD(\cH):=\lbrace\rho\in\cP(\cH)\mid \tr\rho=1\rbrace$ denote the set of density matrices (or states) on $\cH$. We denote the support of an operator $A$ by ${\mathrm{supp}}(A)$ and the range of a projection operator $P$ as ${\mathrm{ran}}(P)$. Let $\mathbb{I}\in\cP(\cH)$ denote the identity operator on $\cH$, and $\id:\cB(\cH)\mapsto \cB(\cH)$ the identity map on operators on~$\cH$.
	Any element $A$ of $\cB_{sa}(\cH)$ has a \textit{spectral decomposition} of the form
	$A = \sum_{\lambda \in \spec(A)} \lambda \,P_\lambda(A),$ where $\spec(A)$ denotes the spectrum of $A$, and $P_\lambda(A)$ is the projection operator corresponding to the eigenvalue $\lambda$. For two superoperators $\Phi_1$ and $\Phi_2$, we denote their composition $\Phi_1\circ\Phi_2$ by $\Phi_1\Phi_2$. We recall that given two $C^*$ algebras of operators $\cA$ and $\cB$, an \textit{operator concave} function $f:\RR\to\RR$ is such that for any two self-adjoint operators $A_1,A_2\in\cA$ and any $\lambda\in [0,1]$:
	\begin{align*}
		f(\lambda A_1+(1-\lambda)A_2)\ge \lambda f(A_1)+(1-\lambda)f(A_2).
		\end{align*}
	A function $f:\RR\to \RR$ is \textit{operator convex} if $-f$ is operator concave. The following  operator generalization of Jensen's inequality will turn out very useful:

		\begin{theorem}[Operator Jensen inequality, see \cite{D57,HP81}]\label{Jensen}
			Let $\mathcal{A}$ and $\mathcal{B}$ be two $C^*$-algebras, and $v:\mathcal{B}\to \mathcal{A}$ a contraction. Then for any operator concave function $f$ on $(0,\infty)$, and any positive element $a\in\mathcal{A}$,
			\begin{align*}
				f(v^*av)\ge v^*(f(a))v, ~~~~~~~~~~~~~~~~~~ \forall a\ge 0.
			\end{align*}	
		\end{theorem}

We use the framework of {\em{relative modular operators}} in our proofs and intermediate results. Relative modular operators were introduced originally by Araki. He used them to extend the notion of relative entropy to pairs of arbitrary states on a C*-algebra (see \cite{Araki76, Araki77, OP}). The relation between relative modular operators and R\'{e}nyi divergences was studied by Petz (see  \cite{Petz1985} and \cite{Petz1986}). Below we briefly recall the definition and basic properties of relative modular operators in the finite-dimensional setting. For more details see e.g.~\cite{DPR16, Jaksicetal}.

\subsubsection*{Relative Modular Operators}

	To define relative modular operators on a finite-dimensional operator algebra $\cB(\cH)$, we start by equipping $\cA\equiv \cB(\cH)$ with a Hilbert space structure through the Hilbert-Schmidt scalar product, which for $A,B \in \cA$ is given by $\langle A,B\rangle := \tr (A^* B)$. We define a map $\pi:\cB(\cH)\to \cB(\cA)$ by $\pi(A):X\mapsto AX$, i.e.~$\pi(A)$ is the map acting on $\cA$ by left multiplication by $A$. The map $\pi$ is linear, one-to-one and has in addition the properties $\pi(AB)=\pi(A)\pi(B)$, and $\pi(A^*)=\pi(A)^*$, where $\pi(A)^*$ denotes the adjoint of the map $\pi(A)$ defined through the 
	relation $\langle X,\pi(A)(Y)\rangle = \langle \pi(A)^*X,Y\rangle$. The following identity between operator norms holds: $\|\pi(A)\|_{\cB(\cA)}=\|A\|_{\cB(\cH)}$. Due to this identity, and the fact that $\pi(A)X=AX$, we identify $A$ with $\pi(A)$ and simply write~$A$ for $\pi(A)$ (even though $\pi(A)$ is a linear map on $\cA$, and $A$ is not!).
	
	For any $\rho\in \cD(\cH)$, we denote $\Omega_\rho:=\rho^{1/2}\in \cB_{sa}(\cH)$. We then have the identity
	\begin{equation}
	\tr(\rho A) = \langle \Omega_\rho, A \Omega_\rho \rangle \qquad \mbox{for all }A\in\cA,\label{eq_GNS}
	\end{equation}
	where the right-hand side of the above identity should be understood as $\langle \Omega_\rho, \pi(A) \Omega_\rho \rangle$. \Cref{eq_GNS} is nothing but a simple case of the so-called GNS representation (see e.g. Section 2.3.3 of \cite{BR1}).\\\\
	For simplicity of exposition, in this paper, we only consider faithful states, i.e.~states~$\rho$ for which $\supp(\rho)=\cH$. Hence, for any pairs of states $\rho, \sigma$, we have $\supp(\rho)=\supp(\sigma)$. We then define the \textit{relative modular operator} $\Delta_{\sigma|\rho}$ to be the map
	\begin{equation}\label{eq_defDelta}
	\begin{array}{cccc}
	\Delta_{\sigma|\rho}: & \cA & \to & \hspace{-1em} \cA \\
	& A & \mapsto & \sigma A \rho^{-1}
	\end{array}
	\end{equation}
	Note that \eqref{eq_defDelta} defines $\Delta_{\sigma|\rho}$ not only for $\rho,\sigma \in \cD(\cH)$, but for any $\rho, \sigma \in \cP_{+}(\cH)$.\\\\	
	As a linear operator on $\mathcal B(\cH)$, $\Delta_{\sigma|\rho}$ is positive and its spectrum $\spec(\Delta_{\sigma|\rho})$ consists of the ratios of eigenvalues $\mu/\lambda$, $ \lambda \in \spec(\rho)$, 
	$ \mu \in \spec(\sigma)$. For any $x\in \spec (\Delta_{\sigma|\rho})$, the corresponding spectral projection is the map
	\begin{equation} \label{eq_specprojDelta}
	\begin{array}{cccc}
	P_{x}(\Delta_{\sigma|\rho}): & \cA & \to & \hspace{-1em} \cA \\
	& A & \mapsto & \underset{\lambda\in \spec(\rho),~\mu\in\spec(\sigma):~\mu/\lambda=x}{\sum}P_\mu(\sigma) A  P_\lambda(\rho).
	\end{array}
	\end{equation}
%	Further, for any $s \in [0,1]$,
%	$\Delta_{\rho|\sigma}^s(A) = \rho^s A \sigma^{-s}$. For any $\rho, \sigma \in \cP_{+}(\cH)$ and $s \in [0,1]$, the quantity $\Psi_s(\rho|\sigma)$ can be expressed in terms of the 
%	relative modular operator $\Delta_{\rho|\sigma}$ as follows
%	\be\label{phi-s-inner}
%	{\Psi_s(\rho|\sigma)} =  \log \langle\Omega_\sigma, \Delta_{\rho|\sigma}^s\Omega_\sigma \rangle, \qquad \text{where }\Omega_{\sigma}:=\sigma^{1/2}.
%	\ee
	By von Neumann's Spectral Theorem (see e.g. Sections VII and VIII of \cite{RS1}) one can associate a classical random variable $X$ to any pair $({\Lambda}, \Omega)$, where $\Lambda$ is a map
$\Lambda : \cA \to \cA$ and $\Omega \in \cB_{sa}(\cH)$, such that for any bounded measurable function $f$, 
\begin{align*}
\langle\Omega, f(\Lambda)\Omega\rangle = {\mathbb{E}}[f(X)] \equiv \int\! f(x) \,\d\mu(x).
\end{align*}
Here $\mu$ denotes the law of $X$ and is referred to as the {\em{spectral measure}} of $\Lambda$ with respect to $\Omega$. For the choice $\Lambda =  \log \Delta_{\sigma|\rho}$ and $\Omega = \Omega_\rho \equiv \rho^{1/2}$,
this yields
	\begin{align} \label{spec-meas}
	\langle\Omega_\rho, f(\log \Delta_{\sigma|\rho})\Omega_\rho\rangle = \int\! f(x) \,\d\mu_{\sigma|\rho}(x) \equiv {\mathbb{E}}[f(X)],
	\end{align}
	where $X$ is a random variable of law $\mu \equiv \mu_{\sigma|\rho}$.
%	We then have in particular
%	\begin{align} \label{eq_fundamental}
%	\Psi_s(\rho|\sigma) = \log \left(\int e^{-sx}\,\d\mu_{\rho|\sigma}(x)\right) \equiv \log {\mathbb{E}}[e^{-sX}]
%	.
%	\end{align}
The relation \eqref{spec-meas} 
	plays a key role in our proofs since it allows us to express the error probabilities of asymmetric hypothesis testing in terms of probability distributions of a classical random variable, 
and therefore allows us to employ the tools of classical probability theory in our analysis. %Another way to associate classical probability distributions to a pair of quantum states is via the so-called Nussbaum-Szko\l a distributions (see~\cite{NussbaumSzkola}).
Taking $f$ to be the identity function, we get:
	\begin{align}
		\EE[X]=\langle \Omega_\rho,\log\Delta_{\sigma|\rho}\Omega_\rho\rangle=- D(\rho\|\sigma),
	\label{2.7}
\end{align}	
where
	\begin{align}
		D(\rho\|\sigma):=\tr\rho(\log\rho-\log\sigma)\label{qrelent}
		\end{align}
is the \textit{quantum relative entropy} of $\rho$ with respect to $\sigma$.
The last identity in \Cref{2.7} can be verified easily by direct computation. Similarly, by taking $f$ to be the square function, one can verify that 
	\begin{align}\label{VV}
		\EE[X^2]-\EE[X]^2 =\langle \Omega_\rho,(\log\Delta_{\sigma|\rho})^2\Omega_\rho\rangle-D(\rho\|\sigma)^2\equiv V(\rho\|\sigma),
	\end{align}	
	where $V(\rho\|\sigma)$ is called the \textit{quantum information variance} and is defined as follows:
	\begin{align}\label{V}
		V(\rho\|\sigma):=\tr\rho(\log\rho-\log\sigma)^2-D(\rho\|\sigma)^2.
		\end{align}

		\subsection*{Conditional expectations and discrete-time martingales}\label{mart}
		 A discrete-time martingale is a sequence of random variables for which, at a particular time in the realized sequence, the expectation of the next value in the sequence is equal to the present observed value, given the knowledge of all prior observed values. More precisely, it is defined as follows. Let $(E,\cF,\PP)$ be a probability space, where $E$ is a set, $\cF$ is a $\sigma$-algebra on $E$ (which is a set of subsets of $E$ containing the empty set, and closed under the operations of taking the complement and discrete unions), and $\PP$ is a probability measure on $\cF$. In the case of a finite set $E$, $\cF$ is usually the set $2^E$ of all the subsets of $\Omega$. Given a measurable space $(E,\cF)$, a \textit{filtration} $\{\cF_n\}_{n\in \NN\cup \{0\}}$ is a sequence of $\sigma$-algebras such that 
		 	\begin{align*}
		 		\cF_0\subseteq \cF_1\subseteq\dots\subseteq \cF_n\dots\subseteq \cF.
		 		\end{align*}
		 		Given a sequence of random variables $\{X_n\}_{n\in\NN\cup\{0\}}$, we denote by $\sigma(X_1,\dots,X_n)$ the smallest $\sigma$-algebra on which the random variables $X_1,\dots,X_n$ are measurable, and call $$\{\sigma(X_1,\dots,X_n)\}_{n\in\NN\cup\{0\}}$$ the  \textit{natural filtration} of $\{X_n\}_{n\in\NN\cup\{0\}}$. More generally a filtration $\{\cF_n\}_{n\in\NN\cup\{0\}}$ is said to be \textit{adapted} to a sequence of random variables $\{X_n\}_{n\in\NN\cup\{0\}}$ if for each $n$, $X_n$ is $\cF_n$-measurable. For a given $\sigma$-algebra $\cG$ on a discrete space, a random variable $X$ is $\cG$-measurable if it can be written as
		 		\begin{align*}
		 			X(\omega)=\sum_{k\in I} x_k\mathbf{1}_{B_k}(\omega),~~~\omega\in\Omega,
		 			\end{align*}
		 			where $I$ is an index set, and $\{B_k\}_{k\in I}$ is a family of disjoint subsets of $\cG$. 
		 		\\\\Consider a sub-$\sigma$-algebra $\cG$ of $\cF$ and an $\cF$-measurable integrable real-valued random variable $X:\Omega\to \RR$, i.e.
		 \begin{align*}
		 	\EE[|X|]:= \int_E |X(\omega)|~\text{d}\PP(\omega)<\infty.
		 \end{align*}	
		 Then the conditional expectation of $X$ with respect to $\cG$ is defined as the almost surely unique (i.e.~up to a set of measure zero) integrable $\cG$-measurable real random variable $Y:=\EE[X|\cG]:E\to \RR$ such that for any other bounded $\cG$-measurable random variable $Z$:
		 \begin{align*}
		 	\EE[ZY]=\EE[ZX].
		 \end{align*}	 
		 In the case of a discrete probability space, the conditional expectation can be expressed as follows: pick any generating family $\{A_k\}_{k\in J}$ of disjoint subsets of $\cG$, with $J$ denoting an index set. Then
		 \begin{align}\label{sumrv}
		 	\EE[X|\cG](\omega)=\sum_{k\in J}  \frac{\EE[X \mathbf{1}_{A_k}]}{\PP(A_k)}~ \mathbf{1}_{A_k}(\omega),\qquad \omega\in E.
 		 \end{align}	
		  The conditional expectation is a linear operation. Moreover, it is easy to verify from \Cref{sumrv} that for any integrable random variable $X$ and sub-$\sigma$-algebra $\cG$, 
		 \begin{align}\label{indmeas}
		 	\EE[X|\cG]=\left\{\begin{aligned}
		 		&X~~~~~~\text{ if $X$ is $\cG$-measurable}\\
		 		& \EE[X] ~~\text{ if $X$ is independent of $\cG$.}
		 		\end{aligned}\right.
		 \end{align}	
		  Let $\{\cF_n\}_{n\in\NN\cup \{0\}}$ be a filtration of $\cF$ and suppose we are given a sequence of real-valued random variables $\{X_n\}_{n\in \NN\cup \{0\}}$ such that for each $n$, $X_n$ is integrable and $\cF_n$-measurable. Then $\{X_n,\cF_n\}_{n\in \NN\cup \{0\}}$ is said to be a \textit{martingale} if for each $n\in\NN\cup \{0\}$,\footnote{a.s.= almost surely}
		 \begin{align*}
		 	\EE[X_{n+1}|\cF_n]=X_n \text{ a.s.}
		 \end{align*}	
		 Similarly, $\{X_n,\cF_n\}_{n\in \NN\cup \{0\}}$ is said to be a \textit{super-martingale} if for each $n\in\NN\cup \{0\}$,
		 \begin{align*}
		 	\EE[X_{n+1}|\cF_n]\le X_n \text{ a.s.}
		 \end{align*}	
		 and a \textit{sub-martingale} if for each $n\in\NN\cup \{0\}$,
		 \begin{align*}
		 	\EE[X_{n+1}|\cF_n]\ge X_n \text{ a.s.}
		 \end{align*}	
	\begin{example}\label{randomwalk}	 Perhaps the simplest example of a martingale is the sum of independent integrable centered random variables. Indeed, let $\{X_n\}_{n\in\NN}$ be such a sequence, $\{\cF_n:=\sigma(X_1,\dots,X_n)\}_{n\in \NN}$ its natural filtration and define $Y_n=\sum_{k=1}^n X_k$. Then
		 \begin{align*}
		 	\EE[Y_{n+1}|\cF_n]&=\sum_{k=1}^{n+1}\EE[X_k|\cF_n]=\sum_{k=1}^{n}\EE[X_k|\cF_n]+\EE[X_{n+1}|\cF_n]\\
	&=\sum_{k=1}^n X_k + \EE[X_{n+1}]= Y_n+0=Y_n,
		 	\end{align*}
		 	where in the first line we used the linearity of the conditional expectation, and in line two we used both identities of \Cref{indmeas}. Therefore $\{Y_n,\cF_n\}_{n\in\NN\cup \{0\}}$ is a martingale, where $Y_0=0$.
		 	\end{example}
	\subsection*{Martingale concentration inequalities} \label{martconc}
 Roughly speaking, the concentration of measure phenomenon can be stated in the following way \cite{TT96}: ``A random variable that depends in a smooth way on many independent random variables (but not too much on any of them) is essentially constant''. This means that such a random variable, $X$, concentrates around its mean (or median) in a way that the probability of the event $\{|X-\EE[X]|>t\}$ decays exponentially in $t\ge0$. For more details on the theory of concentration of measure see \cite{L05}.\\\\
Several techniques have been developed so far to prove concentration inequalities. The method that we focus on here is the martingale approach (see e.g. \cite{BLM13}, \cite{RS13} Chapter 2 and references therein).
The Azuma-Hoeffding inequality has been often used to prove concentration phenomena for discrete-time martingales whose jumps are almost surely bounded. Hoeffding \cite{H63} proved this inequality for a sum of independent and bounded random variables, and Azuma \cite{A67} later extended it to martingales with bounded differences.

	\begin{theorem}[Azuma-Hoeffding inequality\label{AzumaHoeffding}] Let $\{X_k,\cF_k\}_{k\in\NN\cup\{0\}}$ be a discrete-parameter real-valued super-martingale. Suppose that for every $k\in \{1,\dots, n\}$ the condition $|X_k-X_{k-1}|\le d_k$ holds a.s. for a real-valued sequence $\{d_k\}_{k=1}^n$ of non-negative numbers. Then for every $\alpha\ge 0$,
		\begin{align}\label{azumahoeffding}
			\mathbb{P}(X_n-X_0\ge \alpha)\le  \exp\left(-\frac{\alpha^2}{2\sum_{k=1}^nd_k^2}\right).
		\end{align}
	\end{theorem}
	The next result from \cite{MD89} (see also \cite{RS13} Corollary 2.3.2) provides an improvement over the Azuma-Hoeffding inequality  in the limit of large $n$, in the case in which $d_k=d$ for any $k$, by making use of the variance.
	\begin{theorem}\label{tightazu} Let $\{X_k,\cF_k\}_{k\in\NN\cup\{0\}}$ be a discrete-parameter real-valued super-martingale. Assume that, for some constants $0<\nu<d$ the following two inequalities are satisfied almost surely:
		\begin{align*}
			&X_k-\EE[X_k|\cF_{k-1}]\le d\\
			& \EE[(X_k-\EE[X_k|\cF_{k-1}])^2|\cF_{k-1}]\le \nu^2
		\end{align*}
		for every $k\in\{1,\dots,n\}$. Then for every $\kappa \ge 0$,
		\begin{align}\label{improveazuma}
			\PP(X_n-X_0\ge \kappa n)\le \exp\left( -nD_{{bin}}\left(\frac{\delta+\gamma}{1+\gamma}\left\|\frac{\gamma}{1+\gamma}\right.\right)\right),
		\end{align}
		where 
		\begin{align*}
			\gamma:=\frac{\nu^2}{d^2}<1,\qquad \delta:=\frac{\kappa}{d}
		\end{align*}
		and $D_{{bin}}(.\|.)$ here denotes the binary classical relative entropy:
		\begin{align*}
			D_{{bin}}(p\|q):=p\log\left(\frac{p}{q}\right)+(1-p)\log\left(\frac{1-p}{1-q}\right),~~~~p,q\in [0,1].
		\end{align*}
		If $\delta>1$ then these probabilities are equal to zero.
	\end{theorem}	
	To see why \reff{improveazuma} is indeed an improvement over the Azuma-Hoeffding inequality (\ref{azumahoeffding}) in the limit of large $n$ (in the case in which $d_k=d$ for any $k$), use the following identity, which is obtained by a Taylor expansion of $\log(1+u)$: 
	\begin{align*}
		(1+u)\log(1+u)=u+\sum_{k\ge2} \frac{(-u)^k}{k(k-1)},~~~~ -1<u\le 1.
	\end{align*}
	Then it follows that
	\begin{align}
		nD_{{bin}}\left( \frac{\delta+\gamma}{1+\gamma}\left\|\frac{\gamma}{1+\gamma}\right.\right) &=\frac{n}{1+\gamma} \sum_{k\ge 2}\frac{1}{k(k-1)} {\delta}^k\left(1-\left(-\frac{1}{\gamma}\right)^{k-1}\right)
		\nonumber\\
		&= \frac{\delta^2n}{2\gamma}+\mathcal{O}\left(\delta^3 n\right).\label{asymp}
	\end{align}
	The first term leads to an improvement by a factor of $\frac{1}{\gamma}$ over the Azuma-Hoeffding bound (\ref{azumahoeffding}).\\\\
	In the special case of a martingale of the form given in \Cref{randomwalk}, the following concentration inequality was proved by Kearns and Saul \cite{KS98}. It is a refinement of the well-known Hoeffding inequality \cite{H63} and its proof is analogous to the proof of the latter. We employ it in our analysis of quantum hypothesis testing for the case of uncorrelated states (see \Cref{sec_Examples}). Note that for \Cref{randomwalk}, the Hoeffding inequality, and hence also the Kearn-Saul inequality, provide an improvement over the Azuma-Hoeffding inequality.
%	\begin{theorem}[Hoeffding inequality] Let $\{X_k,\cF_k\}$ be independent real-valued bounded random variables (not necessarily i.i.d.), such that for every $k\in \{1,\dots,n\}$, $X_k\in [a_k,b_k]$ holds a.s. for constants $a_k,b_k\in \RR$. Let $\mu_n:=\sum_{k=1}^n \EE[X_k]$. Then for every $\alpha\ge 0$
%			\begin{align}
%				\PP\left(\sum_{k=1}^n X_k -\mu_n \le -\alpha\sqrt{n}\right)\le \e^{-\frac{2\alpha^2 n}{\sum_{k=1}^n (b_k-a_k)^2}}.
%				\end{align}
%	\end{theorem}
	\begin{theorem}\label{KSineq}[Kearns-Saul inequality] Let $\{X_k\}_{k\in\NN\cup\{0\}}$ be independent real-valued bounded random variables, such that for every $k\in \{1,\dots,n\}$, $X_k\in [a_k,b_k]$ holds a.s. for constants $a_k,b_k\in \RR$. Let $\mu_n:=\sum_{k=1}^n \EE[X_k]$. Then for every $\alpha\ge 0$
		\begin{align*}
			\PP\left(\sum_{k=1}^n X_k -\mu_n \le -\alpha\sqrt{n}\right)\le \exp\left(-\frac{\alpha^2 n}{4\sum_{k=1}^n c_k}\right),
		\end{align*}
		where 
		\begin{align*}
			c_k:=\left\{ \begin{aligned}
				&\frac{(1-2p_k)(b_k-a_k)^2}{4\log\left( \frac{1-p_k}{p_k}\right)}~~\text{ if }p_k\ne \frac{1}{2}\\
				& \frac{(b_k-a_k)^2}{8}~~~~~~~~~~~~~\text{ if }p_k=\frac{1}{2},
				\end{aligned}\right.
				\end{align*}
				where $p_k$ is defined as
				\begin{align*}
					p_k:=\frac{\EE[X_k]-a_k}{b_k-a_k},~~~k\in \{1,\dots,n\}.
					\end{align*}
	\end{theorem}	
	This indeed improves Hoeffding's inequality unless $p_k=\frac{1}{2}$ for all $k\in\{1,\dots,n\}$.

\section{Hypothesis testing via martingale methods}\label{sec_main}

\subsection{Finite blocklength analysis of the Type II error exponent}
	Let us fix a sequence of finite dimensional Hilbert spaces $\{\cH_n\}_{n\in \NN}$, and let $\{\rho_n\}_{n\in \NN}$ and $\{\sigma_n\}_{n\in \NN}$ denote two sequences of states, where for each $n\in\NN$, $\rho_n, \sigma_n \in \cD(\cH_n)$. For a test $0\le T_n\le \mathbb{I}_n $, the Type I and Type II errors for the corresponding binary quantum hypothesis testing problem are given by
	\begin{align*}
	\alpha(T_n)=\tr[(\mathbb{I}_n-T_n)\rho_n]\qquad \text{ and }\qquad \beta(T_n)=\tr[T_n\sigma_n].
\end{align*}	
	As mentioned in the introduction, in the context of asymmetric hypothesis testing, the two quantities of interest are the {\em{Stein error}} and the {\em{Hoeffding error}}, defined through \Cref{qSteinII} and \Cref{qHoeffII}, respectively. In this section we obtain bounds on these errors for finite blocklength, i.e.~for finite values of $n$, for uncorrelated states, that is when $\rho_n$ and $\sigma_n$ are each given by a tensor product of $n$ (not necessarily identical) states.
		\begin{remark}
			We restrict our consideration to the case of faithful states $\rho_n$, $\sigma_n$ only to make our exposition more transparent. Simple limiting arguments show that all our results remain valid in the case in which $\supp(\rho_n) \subseteq \supp(\sigma_n)$.
		\end{remark}In fact, our upper bounds on the Stein- and Hoeffding errors, as given in \Cref{mainresult}, are valid when the sequences $\{ \rho_n\}_{n\in\NN}$ and $\{ \sigma_n\}_{n\in\NN}$ satisfy \Cref{cond-sub-mart} given below. 

	\begin{condition}\label{cond-sub-mart}
	The states $\rho_n, \sigma_n \in \cD(\cH_n)$ of the sequences $\{ \rho_n\}_{n\in\NN}$ and $\{ \sigma_n\}_{n\in\NN}$ are such that the random variables $Y_0=0$ and $Y_n:= X_n+D(\rho_n\|\sigma_n)$, 
where $X_n$ is the random variable associated to the pair $(\log \Delta_{\sigma_n|\rho_n}, \Omega_{\rho_n})$ through \Cref{spec-meas}, form a super-martingale with respect to their natural filtration. Moreover, there exists a sequence $\{d_k\}_{k\in\NN}$ of non-negative numbers such that for any $k\ge 1$, $|Y_k-Y_{k-1}|\le d_k$ almost surely i.e.~with probability $1$.
	\end{condition}
	\begin{remark}
		One can readily verify that $\EE[X_n]=-D(\rho_n\|\sigma_n)$, so that $Y_n$ is a centered random variable, for each $n\in\NN$.
	\end{remark}	
	As shown below, uncorrelated states satisfy the above condition. Later in the paper, we show how a refined analysis allows us to recover similar results for certain classes of correlated states, i.e. those satisfying a so-called factorization property (see Sections \ref{corr} and \ref{finitecorrelated}).\\\\
	Our upper bounds on the finite blocklength Stein- and Hoeffding errors are stated in the following lemma: 
	\begin{lemma}[Upper bounds on finite blocklength optimal asymmetric error exponent]\label{mainresult}Let $\{\rho_n\}_{n\in\NN}$ and $\{\sigma_n\}_{n\in\NN}$ be two sequences of states that satisfy \Cref{cond-sub-mart}. Then for any $\eps>0$ there exists a sequence of tests $\{T^{\eps}_n\}_{n\in\NN}$ such that for any $n\in \NN$,
		\begin{align*}
	\alpha(T^\eps_n)\le \eps,\qquad
		\beta(T^\eps_n)\le \exp\left({-D(\rho_n\|\sigma_n)+\sqrt{2\log(1/\eps) \sum_{k=1}^n d_k^2}}\right).
		\end{align*}
		Moreover, for any $r>0$, there exists a sequence of tests $\{\tilde{T}_n^{r}\}_{n\in\NN}$ such that for each $n\in\NN$,
		\begin{align*}
						\alpha(\tilde{T}^{r}_n)\le \e^{-nr},
				\qquad \beta(\tilde{T}^r_n)\le \exp\left({-D(\rho_n\|\sigma_n)+\sqrt{2nr \sum_{k=1}^n d_k^2}}\right).
				\end{align*}
		Hence, for each $n\in\NN$,
		\begin{align}
		&	\beta_n(\eps)\le \exp\left(-D(\rho_n||\sigma_n)+\sqrt{2\sum_{k=1}^n d_k^2 \log(1/\eps)}\right)\label{qstein-thm4}\\
		&				\tilde{\beta}_n(r)\le \exp\left(-D(\rho_n||\sigma_n)+\sqrt{2nr\sum_{k=1}^n d_k^2 }\right).\label{hoeffding-thm4}
		\end{align}	
	\end{lemma}	
	In order to prove \Cref{mainresult}, we use the Azuma-Hoeffding martingale concentration inequality (\Cref{AzumaHoeffding}) as well as the following result, which allows us to relate the error probabilities arising in asymmetric quantum hypothesis testing to laws of classical super-martingales. 
The latter result was stated as Proposition 1 in \cite{DPR16} but its proof is essentially due to Li \cite{L14}.

	\begin{proposition}\label{prop_1}\cite{DPR16}
		Let $\rho$, $\sigma$ be two states in ${\mathcal{D}}({\mathcal{H}})$. For any $L>0$ there exists a test $T$ such that
		\begin{equation} 
			\tr \,\rho(1-T) \leq  \big\langle{\Omega_\rho}, P_{[-\log L,\infty)} (\log \Delta_{\sigma|\rho})\, \Omega_\rho\big\rangle \qquad
			\mbox{and}\qquad \tr\,\sigma T \leq  L^{-1}, \label{eq_alphabeta}
		\end{equation}
		where $ P_{[-\log L,\infty)} (\log \Delta_{\sigma|\rho}) := \sum_{x \in [-\log L,\infty)} P_x (\log \Delta_{\sigma|\rho})$, with $ P_x (\log \Delta_{\sigma|\rho})$ being the spectral projection operator of $\log \Delta_{\sigma|\rho}$ of associated eigenvalue $x$.
	\end{proposition}
	
For the proof of this proposition, see \cite{DPR16}. The proof actually provides a construction of the tests $\{T_n^{\eps}\}_{n\in\NN}$ and $\{\tilde{T}_n^r\}_{n\in\NN}$ appearing in \Cref{mainresult}.

	\subsection*{Proof of \Cref{mainresult}} \label{sec_proofmainthm}
For $n\in\NN$, fix $0< L_n\le\e^{D(\rho_n\|\sigma_n)}$. Then by \reff{eq_alphabeta}, there exists a test $T_n$ such that
	\begin{align*}
	\alpha(T_n)\le \langle \Omega_{\rho_n}, P_{[-\log L_n,\infty)}(\log \Delta_{\sigma_n|\rho_n})(\Omega_{\rho_n})\rangle\qquad\text{ and }\qquad\beta(T_n)\le L_n^{-1}.
	\end{align*}	
Now
	\begin{align}
		\langle \Omega_{\rho_n}, P_{[-\log L_n,\infty)}(\log \Delta_{\sigma_n|\rho_n})(\Omega_{\rho_n})\rangle &=\langle \Omega_{\rho_n}, {\indc}_{[-\log L_n,\infty)}(\log \Delta_{\sigma_n|\rho_n})(\Omega_{\rho_n})\rangle\nonumber\\ 
&= {\mathbb{E}}\left({\indc}_{[-\log L_n,\infty)}(X_n)\right)\nonumber\\
&=\PP(X_n\ge- \log L_n)=\PP(Y_n\ge -\log L_n+D(\rho_n\|\sigma_n)).\nonumber
	\end{align}
	where $X_n$ is the random variable associated to the pair $(\log\Delta_{\sigma_n|\sigma_n}, \Omega_{\rho_n})$, and $Y_n:=X_n+D(\rho_n\|\sigma_n)$. Assuming that \Cref{cond-sub-mart} is satisfied, an application of \Cref{AzumaHoeffding} to the super-martingale $\{Y_k,\cF_k\}_{k\in\NN\cup\{0\}}$ with $Y_0=0$, where $\{\cF_k\}_{k\in\NN}$ is the natural filtration associated with the random variables $X_k$, yields the following:
	\begin{align}
		\alpha(T_n)&\le \PP\Bigl( Y_n-Y_0\ge  D(\rho_n\|\sigma_n)-\log L_n \Bigr) \nonumber\\
		&\le  \exp\left({-\frac{(D(\rho_n\|\sigma_n)-\log L_n)^2}{2 \sum_{k=1}^n d_k^2}}\right).
\label{3.10}
		\end{align}
		Setting the quantity on the right hand side of the above inequality to be equal to $\eps$, and using the fact that $\log L_n\le {D(\rho_n\|\sigma_n)}$, we find that 
		\begin{align*}
			\log L_n = D(\rho_n\|\sigma_n) -\sqrt{2\log(1/\eps) \sum_{k=1}^n d_k^2}.
			\end{align*}
		This implies that
		\begin{align*}
			\beta(T_n)\le L_n^{-1}= \exp\left({-D(\rho_n\|\sigma_n)+\sqrt{2\log(1/\eps) \sum_{k=1}^n d_k^2}}\right),
			\end{align*}
			from which \reff{qstein-thm4} follows since $\beta_n(\eps)\le \beta(T_n)$. The inequality \reff{hoeffding-thm4} can be derived analogously by following the same steps 
as above but replacing $\eps$ by $\e^{-nr}$.

\qed

\subsection{A lower bound on the second order asymptotics of the Type II error exponent}		
As yet another application of a martingale concentration inequality in quantum hypothesis testing, we obtain a lower bound on the second order asymptotics of the type II error exponent, $-\log \beta_n(\eps)$, for the case in which the states $\rho_n, \sigma_n$ occurring in the sequence satisfy the more constrained \Cref{cond2}. The lower bound is given in \Cref{prop2}. In particular, \Cref{cond2} can be readily verified to be satisfied when $\rho_n$ and $\sigma_n$ are of the tensor product form.
\begin{condition}\label{cond2} 
		The states $\rho_n$ and $\sigma_n$ of states on each $\cH_n$ are such that the random variables $Y_0=0$ and $Y_n:= X_n+D(\rho_n\|\sigma_n)$, where $X_n$ is the random variable associated to the
pair $(\log \Delta_{\sigma_n|\rho_n}, \Omega_{\rho_n})$, form a super-martingale with respect to their natural filtration. Moreover, assume that for some constants $d$ and $\nu$ the following two requirements are satisfied almost surely:
		\begin{align*}
			&Y_k-\EE[Y_k|\cF_{k-1}]\le d\\
			& \EE[(Y_k-\EE[Y_k|\cF_{k-1}])^2|\cF_{k-1}]\le \nu^2.
		\end{align*}
\end{condition}

\begin{proposition}\label{prop2}
Suppose that the sequences of states $\{\rho_n\}_{n\in\NN}$ and $\{\sigma_n\}_{n\in\NN}$ satisfy \Cref{cond2}. Then for any sequence $\{L_n\}_{n\in\NN}$ of positive numbers such that for any $n\in\NN$, $D(\rho_n\|\sigma_n)\ge \log L_n$, there exists a sequence of tests $\{T_n\}_{n\in\NN}$ such that for any $n\in\NN$ the type I and type II errors satisfy the following inequalities:
	\begin{align}
		&\alpha(T_n)\le \exp\left( -nD_{{bin}}\left(\frac{\delta_n+\gamma}{1+\gamma}\left\|\frac{\gamma}{1+\gamma}\right)\right.\right),\nonumber\\
		&\beta(T_n)\le L_n^{-1},\label{beta}
		\end{align}
		where $\gamma={\nu^2}/{d^2}$ and for each $n$, $\delta_n=\left(D(\rho_n\|\sigma_n)-\log L_n\right)/{nd}$. This implies that for any $0<\eps<1$:
\begin{align}\label{soa2}
	- \log\beta_n(\eps)\ge D(\rho_n\|\sigma_n) - \sqrt{2n\log(\eps^{-1})}\nu +\mathcal{O}(1).
\end{align}
		\end{proposition}
	\begin{proof} 
	The first part of the proof of this proposition is similar to the proof of \Cref{mainresult} and follows from a simple use of \Cref{tightazu} as well as \Cref{prop_1}. Using \reff{asymp}, one derives the following asymptotic upper bound for $\alpha(T_n)$:
\begin{align*}
	\alpha(T_n)\le\exp\left( -\frac{(D(\rho_n\|\sigma_n)-\log L_n)^2}{n} \frac{1}{2\nu^2}+\mathcal{O}\left(\frac{(D(\rho_n\|\sigma_n)-\log L_n)^3}{n^{2}}\right)\right).
	\end{align*}
Fix $0<\eps<1$. Choosing $\log L_n= D(\rho_n\|\sigma_n)-\sqrt{2n\log \eps^{-1}}\nu$, the last inequality can be simplified:
\begin{align*}
	\alpha(T_n)\le\exp\left(\log\eps +\mathcal{O}\left(\frac{1}{\sqrt{n}}\right)\right)=\eps+ \mathcal{O}\left(\frac{1}{\sqrt{n}}\right).
	 \end{align*}
	  This implies, by a suitable use of Taylor expansion, that 
	  \begin{align*}
	  	- \log\beta_n(\eps)\ge D(\rho_n\|\sigma_n)-\sqrt{2n\log(\eps)^{-1}}\nu +\mathcal{O}(1).
	  	\end{align*}
	  	\qed
		\end{proof}

\subsection{Example: the case of uncorrelated quantum states} \label{sec_Examples}

In this section we consider the case in which $\{\rho_n\}_{n\in\NN}$ and $\{\sigma_n\}_{n\in\NN}$ are sequences of independent (i.e.~uncorrelated) states. We show that in this case \Cref{cond-sub-mart} holds, and hence \Cref{mainresult} can be applied. We also show that in this case, a tighter concentration inequality than the Azuma-Hoeffding inequality of \Cref{AzumaHoeffding} provides better upper bounds on
the Stein- and Hoeffding errors.\\\\Suppose that the sequences $\{\rho_n\}_{n\in\NN}$ and $\{\sigma_n\}_{n\in\NN}$ are such that for each $n$, $\rho_n$ and $\sigma_n$ are of the following tensor product form:
\begin{align*}
	\rho_n=\tilde{\rho}_1\otimes\cdots \otimes \tilde{\rho}_n\qquad\text{ vs }\qquad 	\sigma_n=\tilde{\sigma}_1\otimes\cdots\otimes \tilde{\sigma}_n,
	\end{align*}
	where for each $k\in \{1,2,\ldots, n\}$, $\tilde{\rho}_k,\tilde{\sigma}_k\in \cD(\tilde{\cH}_k)$, where  $\tilde{\cH}_k$ is a finite-dimensional Hilbert space. In this case, 
\begin{align*}
	\log\Delta_{\sigma_n|\rho_n}=\sum_{k=1}^n \id^{\otimes k-1}\otimes \log\Delta_{\tilde{\sigma}_k|\tilde{\rho}_k} \otimes \id^{\otimes n-k}.
	\end{align*}
	As all the terms in the sum in the above identity commute, one finds by functional calculus that for any $u\in\RR$,
	\begin{align*}
		\e^{iu\log\Delta_{\sigma_n|\rho_n}}=\prod_{k=1}^n \exp\left({iu ~\id^{\otimes k-1}\otimes \log\Delta_{\tilde{\sigma}_k|\tilde{\rho}_k} \otimes \id^{\otimes n-k}}\right).
		\end{align*}
		This in turn implies that the random variable $X_n$ has characteristic function
		\begin{align*}
			\EE\left[\e^{iu X_n}\right]=\langle  \Omega_{\rho_n}, \e^{iu \log \Delta_{\sigma_n|\rho_n}}(\Omega_{\rho_n})\rangle&= \prod_{k=1}^n \langle \Omega_{\tilde{\rho}_k} , \e^{iu \log\Delta_{\tilde{\rho}_k|\tilde{\sigma}_k}}(\Omega_{\tilde{\rho}_k})\rangle=\prod_{k=1}^n \EE\left[ \e^{iu \tilde{X}_k} \right],
	\end{align*}
	where for each $1\le k\le n$, $\tilde{X_k}$ is a random variable associated to the pair $(\log\Delta_{\tilde{\sigma}_k|\tilde{\rho}_k}, \Omega_{\tilde{\rho}_k})$. The random variable $X_n$ has the same distribution as the sum 
$X_n = \sum_{k=1}^n \tilde{X}_k$ of independent random variables $\tilde{X}_k$, which in turn implies that the random variable $X_n-\EE[X_n]=X_n+D(\rho_n\|\sigma_n)$ has the same distribution as
	\begin{align*}
		\sum_{k=1}^n \tilde{X}_k -\EE[\tilde{X}_k]=\sum_{k=1}^n \tilde{X}_k +D(\tilde{\rho}_k\|\tilde{\sigma}_k).
	\end{align*}	
	Hence, without loss of generality, the random variable $Y_n$ appearing in \Cref{cond-sub-mart} is equal to $\sum_{k=1}^n \tilde{X_k}+D(\tilde{\rho}_k\|\tilde{\sigma}_k)$. Now $\{Y_n,\cF_n\}_{n\in\NN\cup \{0\}}$ is a martingale of the form of \Cref{randomwalk}, where $\cF_n=\sigma(\tilde{X}_k,k\le n)$. Moreover the random variables $Y_k-Y_{k-1}:=\tilde{X}_k+D(\tilde{\rho}_k\|\tilde{\sigma}_k)$ are bounded by 
	\begin{align}\label{dk}
	d_k:=	\|\log\Delta_{\tilde{\sigma}_k|\tilde{\rho}_k} + D(\tilde{\rho}_k\|\tilde{\sigma}_k) \|_{\infty}.
		\end{align}
This analysis leads to the following corollary of \Cref{mainresult}:
\begin{theorem}[Upper bounds for uncorrelated states]\label{uncorcase}
States of the form	
\begin{align}\label{iidstates}
\rho_n=\tilde{\rho}_1\otimes\cdots \otimes\tilde{\rho}_n\qquad\text{ and }\qquad 	\sigma_n=\tilde{\sigma}_1\otimes\cdots\otimes \tilde{\sigma}_n
\end{align}
on a Hilbert space $\tilde{\cH}_1\otimes\cdots\otimes\tilde{\cH}_n$ satisfy \Cref{cond-sub-mart}. Therefore, they satisfy the bounds given in \reff{qstein-thm4} and \reff{hoeffding-thm4} on the error exponents of type II with coefficients $d_k$ given by \Cref{dk}. More precisely, for any $\eps>0$ there exists a sequence of tests $\{T^{\eps}_n\}_{n\in\NN}$ such that for any $n\in \NN$
	\begin{align*}
		\alpha(T^\eps_n)\le \eps,\qquad
		\beta(T^\eps_n)\le \exp\left(-D(\rho_n\|\sigma_n)+\sqrt{2\log(1/\eps) \sum_{k=1}^n d_k^2}\right).
	\end{align*}
	Similarly, for any $r>0$, there exists a sequence of tests $\{\tilde{T}_n^{r}\}_{n\in\NN}$ such that for each $n\in\NN$:
	\begin{align*}
		\alpha(\tilde{T}^{r}_n)\le \e^{-nr},
		\qquad \beta(\tilde{T}^r_n)\le \exp\left(-D(\rho_n\|\sigma_n)+\sqrt{2nr \sum_{k=1}^n d_k^2}\right),
	\end{align*}
	This implies that for each $n\in\NN$:
	\begin{align}
		&	\beta_n(\eps)\le \exp\left(-D(\rho_n||\sigma_n)+\sqrt{2\sum_{k=1}^n d_k^2 \log(1/\eps)}\right)\label{stein}\\
		&				\tilde{\beta}_n(r)\le \exp\left(-D(\rho_n||\sigma_n)+\sqrt{2nr\sum_{k=1}^n d_k^2 }\right).\label{hoeffding}
	\end{align}	
\end{theorem}	
In this special case of a sum of independent random variables, we can use the tighter concentration inequality, namely the Kearns-Saul inequality (\Cref{KSineq}) to get better bounds than the ones given in \Cref{uncorcase}. This yields:
\begin{theorem}[Improved upper bounds for uncorrelated states]\label{theorem6}
	Let $\{\rho_n\}_{n\in\NN}$ and $\{\sigma_n\}_{n\in\NN}$ two sequences of states of the form given in \Cref{iidstates}. Then  for any $\eps>0$ there exists a sequence of tests $\{T^{\eps}_n\}_{n\in\NN}$ such that, for any $n\in \NN$,
	\begin{align*}
		\alpha(T^\eps_n)\le \eps,\qquad
		\beta(T^\eps_n)\le \exp\left(-D(\rho_n\|\sigma_n)+\sqrt{4\log(1/\eps) \sum_{k=1}^n c_k}\right),
	\end{align*}
	where the constants $c_k$ are given by
		\begin{align}
			c_k:=\left\{ \begin{aligned}\label{ck}
				&\frac{(1-2p_k)(b_k-a_k)^2}{4\log\left( \frac{1-p_k}{p_k}\right)}~~\text{ if }p_k\ne \frac{1}{2}\\
				& \frac{(b_k-a_k)^2}{8}~~~~~~~~~~~~~\text{ if }p_k=\frac{1}{2},
			\end{aligned}\right.
		\end{align}
		with $a_k,b_k$ and $p_k$ defined as
		\begin{align*}
			a_k:= \log\left(\frac{\lambda_{\max}(\tilde{\sigma}_k)}{\lambda_{\min}(\tilde{\rho}_k)}\right),\quad
			b_k:= \log\left(\frac{\lambda_{\min}(\tilde{\sigma}_k)}{\lambda_{\max}(\tilde{\rho}_k)}\right),\quad
&			p_k:=\frac{-D(\tilde{\rho}_k\|\tilde{\sigma}_k)-a_k}{b_k-a_k}.
		\end{align*} 
Similarly, for any $r>0$, there exists a sequence of tests $\{\tilde{T}_n^{r}\}_{n\in\NN}$ such that for each $n\in\NN$:
	\begin{align*}
		\alpha(\tilde{T}^{r}_n)\le \e^{-nr},
		\qquad \beta(\tilde{T}^r_n)\le \e^{-D(\rho_n\|\sigma_n)+\sqrt{4nr \sum_{k=1}^n c_k}},
	\end{align*}
 This implies that for each $n\in\NN$:
	\begin{align}
		&	\beta_n(\eps)\le \exp\left({-D(\rho_n||\sigma_n)+\sqrt{4\log(\eps^{-1})\sum_{k=1}^n c_k }}\right)\label{stein2}\\
		&				\tilde{\beta}_n(r)\le \exp\left({-D(\rho_n||\sigma_n)+\sqrt{4nr\sum_{k=1}^n c_k }}\right).\label{hoeffding2}
	\end{align}

\end{theorem}

\subsubsection{Reduction to the i.i.d. case}\label{sec:iid}
\Cref{uncorcase} gives upper bounds on the Stein- and Hoeffding errors for uncorrelated states. In the case of i.i.d.~states ($\rho^{\otimes n}$ vs.~$\sigma^{\otimes n}$), the bound corresponding to the Stein error \reff{stein} reduces to
\begin{align}\label{finite}
	\frac{1}{n}\log \beta_n(\eps)\le -D(\rho\|\sigma) +\frac{1}{\sqrt{n}} h(\eps),
\end{align}	
where
\begin{align}\label{h}
	h(\eps):=\sqrt{2 \log\eps^{-1}}	\|\log\Delta_{{\sigma}|{\rho}} + D({\rho}\|\sigma) \|_{\infty}
\end{align}	
Taking the limit on both sides, the result is in accordance with quantum Stein's lemma:
\begin{align*}
	\lim_{n\to \infty}-\frac{1}{n}\beta_n(\eps)=D(\rho\|\sigma).
	\end{align*}
As mentioned in the introduction, in the i.i.d. case, the authors of \cite{L14,TH13} moreover showed that
\begin{align}\label{second}
-	\log \beta_n(\eps)=nD(\rho\|\sigma)+\sqrt{n V(\rho\|\sigma)}\Phi^{-1}(\eps)+O(\log{n})
\end{align}	
	where $V(\rho\|\sigma)$ is the quantum information variance defined in \Cref{V}, and $\Phi$ is the cumulative distribution function of a standard Gaussian random variable. In fact, Li's proof is very similar to our proof of \Cref{mainresult}, as he was the first to prove an i.i.d. version of \Cref{prop_1} which was later on adapted to fit non i.i.d. settings in \cite{DPR16}. This result coupled to a Central limit theorem (or more precisely a refinement of it called the Berry Esseen theorem) were the two main ingredients of the proof of \reff{second}. This equation shows that the correct order of the deviation of $\frac{1}{n}\log\beta_n(\eps)$ from $-D(\rho\|\sigma)$ is indeed $\frac{1}{\sqrt{n}}$ (at least for $\eps\ne 1/2$). More precisely, \reff{second} implies that
	 \begin{align}\label{asympt}
	 	\limsup_{n\to\infty} \sqrt{n}\left(\frac{1}{n}\log\beta_n(\eps)+D(\rho\|\sigma)\right)= -\sqrt{V(\rho\|\sigma)}\Phi^{-1}(\eps)
	 \end{align}
	 This implies that for all $\eps\in [0,1]$ and all $1\ge\eps^\prime>\eps$ there exist infinitely many $n\in\NN$ for which
	 \begin{align*}
	 	\sqrt{n}\left( \frac{1}{n}\log\beta_n(\eps)+D(\rho\|\sigma)\right)\ge -\sqrt{V(\rho\|\sigma)}\Phi^{-1}(\eps^\prime)
	 \end{align*}	
	 Using \reff{finite} this imposes that for all $\eps^\prime>\eps$
	 \begin{align}\label{neg}
	 	\sqrt{2\log\eps^{-1}}\|\log\Delta_{\sigma|\rho} + D(\rho\|\sigma) \|_{\infty}\ge -\sqrt{V(\rho\|\sigma)}\Phi^{-1}(\eps^\prime).
	 	\end{align} 
	 As pointed out in \cite{AMV12}, the negative sign on the right hand side of \reff{neg} is justified by the fact that, for 
	 $\eps<1/2$, the right hand side is itself negative. 	 Moreover, to prove \reff{neg} it suffices to show that for any $0<\eps<1/2$,
	   \begin{align*}
	   	\frac{\|\log\Delta_{{\sigma}|{\rho}} + D({\rho}\|{\sigma}) \|_{\infty}}{\sqrt{V(\rho\|\sigma)}}\ge -\frac{\Phi^{-1}(\eps)}{\sqrt{2\log\eps^{-1}}}.
	   \end{align*} 
	 Let us focus on the right hand side of the above equation. By the Gaussian concentration inequality (see e.g. \cite{BLM13} Theorem 5.6),
	 \begin{align*}
	 \Phi(x)\le \e^{-x^2/2},
	 \end{align*}
	setting $x = \Phi^{-1}(\eps)$, we hence infer that for any $0<\eps<1/2$,
	\begin{align*}
		\eps \le \exp\left(-(\Phi^{-1}(\eps))^2/2\right),
		\end{align*}
which in turn implies that
			\begin{align}\label{philog}
				-\Phi^{-1}(\eps) \le \sqrt{2\log\eps^{-1}}.
				\end{align}
		Hence, the inequality (\ref{neg}) holds provided
		\begin{align*}
			\sqrt{V(\rho\|\sigma)}\le \|\log\Delta_{{\sigma}|{\rho}} + D({\rho}\|{\sigma}) \|_{\infty},
		\end{align*}	
		which can easily be verified as follows:
		\begin{align*}
			V(\rho\|\sigma)&=\langle\Omega_\rho,(\log\Delta_{\sigma|\rho}+D(\rho\|\sigma)\id)^2(\Omega_\rho)\rangle\\
			&\le \|\log\Delta_{\sigma|\rho}+D(\rho\|\sigma)\id\|^2_{\infty}
			\end{align*}
Our theorems should finally be compared with the results of \cite{AMV12} where similar bounds have been derived in the i.i.d. setting using different techniques. For example, as already mentioned in the introduction, for the Stein error, it was shown in \cite{AMV12} that (see Theorem 3.3 and Equation (35) of \cite{AMV12}):
	\begin{align}\label{bounds}
	-D(\rho\|\sigma)-\frac{f(\eps)}{\sqrt{n}}	\le\frac{1}{n}\log{\beta_{n}(\eps)}\le -D(\rho\|\sigma)+\frac{g(\eps)}{\sqrt{n}},
	\end{align}	
	where 
	\begin{align}\label{fg}
		f(\eps)= 4\sqrt{2}\log\eta \log(1-\eps)^{-1},\qquad g(\eps)=4\sqrt{2} \log \eta \log\eps^{-1}
		\end{align}
		and $\eta:=1+\e^{1/2D_{3/2}(\rho\|\sigma)}+\e^{-1/2D_{1/2}(\rho\|\sigma)}$.
		The upper bound in \reff{bounds} was found via semidefinite programming, the use of a bound on the optimal error probability in symmetric hypothesis testing in terms of the $\alpha$-R\'{e}nyi divergence (originally derived in \cite{Aetal07}), and an inequality relating the $\alpha$-R\'{e}nyi divergence of the states $\rho$ and $\sigma$ to their relative entropy. The lower bound was derived using the monotonicity of the $\alpha$-R\'{e}nyi divergence under completely positive trace-preserving (CPTP) maps, the CPTP map here being the measurement channel associated to any given test. \\\\
Our bound, given in \reff{finite}, is tighter than the corresponding upper bound obtained in \cite{AMV12} (given by \reff{bounds} and \reff{fg}) for  
 $\eps\le \eps_0$, where
	\begin{align*}
	 \eps_0:=\exp\left(-\frac{\|\log\Delta_{\sigma|\rho}+D(\rho\|\sigma) \|_\infty^2}{16(\log\eta)^2}\right),
	\end{align*}			
 as the dependence of $h$ on $\eps$ is given by the square root of $\log\eps^{-1}$, whereas $g$ behaves as $\log\eps^{-1}$. 
\smallskip

In fact, \Cref{theorem6} yields a bound which is tighter than \reff{finite} and is given by
	\begin{align*}
		\frac{1}{n}\log\beta_n(\eps)\le -D(\rho\|\sigma)+\frac{\tilde{h}(\eps)}{\sqrt{n}},
		\end{align*}
where 
\begin{align}\label{tilde}
	\tilde{h}(\eps):=\sqrt{4c\log(\eps^{-1})},
\end{align}		
and $c$ is defined as follows:
	\begin{align}\label{cconstant}
	 	c:=\left\{ \begin{aligned}
			&\frac{(1-2p)(b-a)^2}{4\log\left( \frac{1-p}{p}\right)}~~\text{ if }p\ne \frac{1}{2}\\
			& \frac{(b-a)^2}{8}~~~~~~~~~~~~~\text{ if }p=\frac{1}{2},
		\end{aligned}\right.
	\end{align}
	with 
	\begin{align*}
		a:= \log\left(\frac{\lambda_{\max}({\sigma})}{\lambda_{\min}({\rho})}\right),\quad
		b:= \log\left(\frac{\lambda_{\min}({\sigma})}{\lambda_{\max}({\rho})}\right),\quad
					p:=\frac{-D({\rho}\|{\sigma})-a}{b-a}.
	\end{align*} 

		The above bound is tighter than \reff{bounds} for all $\eps\le \tilde{\eps}_0$, where 
				\begin{align*}
				 \tilde{\eps}_0&:=\exp\left(-\frac{c^2}{8(\log\eta)^2}\right),
				\end{align*}			
			where $\eta$ is given below \reff{fg}. \\\\
			Finally, one can also compare the asymptotic bound of \reff{soa2} for the i.i.d.~case to \reff{second}. In this case, with the notations of \Cref{cond2},
	\begin{align*}
		\EE[(Y_k-\EE[Y_k|\cF_{k-1}])^2|\cF_{k-1}]&=\EE[(Y_k-Y_{k-1})^2|\cF_{k-1}]=\EE[(\tilde{X}_k+D(\rho\|\sigma))^2|\cF_{k-1}]\\
		&=\EE[(\tilde{X}+D(\rho\|\sigma))^2]=\EE[\tilde{X}^2]-D(\rho\|\sigma)^2=V(\rho\|\sigma).
		\end{align*}
where $\tilde{X}$ and the $\tilde{X}_k$'s are independent, identically distributed random variables of law $\mu_{\sigma|\rho}$, the last identity arising from \Cref{VV}. Hence, \reff{soa2} holds with $\nu=\sqrt{V(\rho\|\sigma)}$ and can be expressed as follows:
\begin{align}\label{secondnewnew}
\frac{1}{n}	\log\beta_n(\eps)\le-D(\rho\|\sigma)+\frac{\sqrt{2\log(\eps^{-1})V(\rho\|\sigma)}}{\sqrt{n}}+\mathcal{O}\left(\frac{1}{n}\right).
	\end{align}
From \reff{philog},
\begin{align}\label{s1s2}
	s_2(\eps):=\sqrt{2\log(\eps^{-1})V(\rho\|\sigma)}\ge -\Phi^{-1}(\eps)\sqrt{V(\rho\|\sigma)}\equiv s_1(\eps),
	\end{align}
	which implies that for $n$ large enough, the asymptotic bound (\ref{secondnewnew}) is looser than the one given by \reff{soa}. \Cref{fig1} shows an example for which our bounds are significantly 
closer to the second-order asymptotic behaviour (given by $s_1(\eps)$) than the original upper bound of \cite{AMV12} for 
$\eps<\tilde{\eps}_0,\eps_0$.

	\begin{figure}[H]
		\centering
		\includegraphics[width=0.80\textwidth]{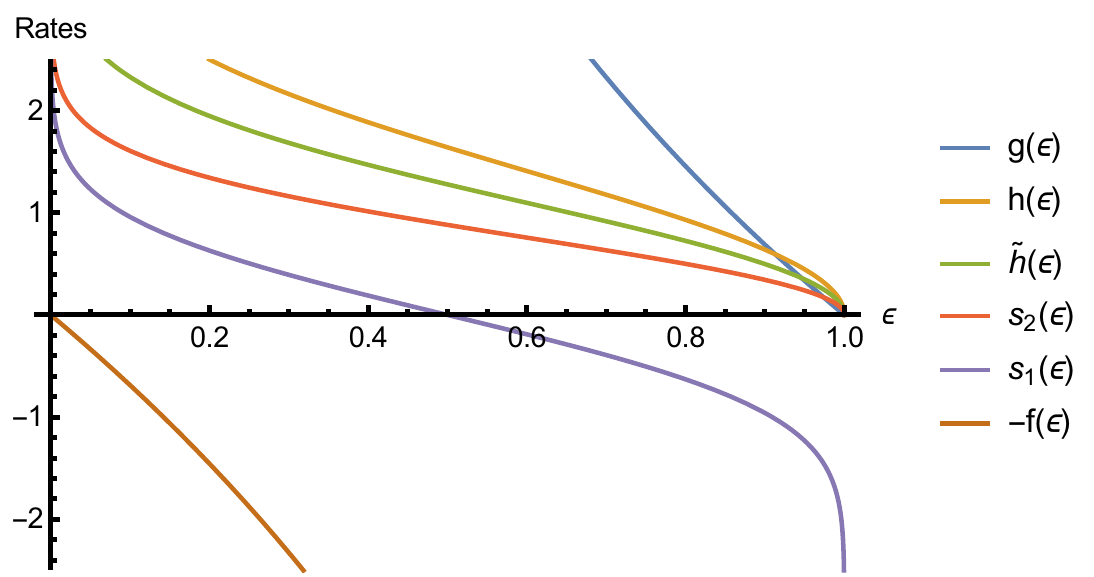}
		\caption{Bounds on the function $Q(n, \eps)$ (given by \Cref{qne}) for a pair of randomly generated qubit states $\rho$ and $\sigma$; here $\rho$ has Bloch vector $(-0.177483, 0.365807, 0.291007
			)$ whereas $\sigma$ has Bloch vector $(-0.452239, -0.141906, -0.159193
			)$. The functions $g$ and $f$ (defined in \reff{fg}) correspond to the finite blocklength bounds found in \cite{AMV12}; $h$ (defined in \reff{h}) corresponds to the bound given in \Cref{uncorcase}, and $\tilde{h}$ (defined in \reff{tilde}) corresponds to the tighter bound given in \Cref{theorem6}. The functions $s_1$ and $s_2$ (cf. \reff{s1s2}) correspond to the second-order asymptotic behaviours stated in \Cref{second} and in \reff{secondnewnew}, respectively.}\label{fig1}
	\end{figure}

	\section{Correlated states and factorization properties}\label{corr}
	
	In the last section, we derived bounds on the optimal type II errors in asymmetric hypothesis testing for uncorrelated states. In the reminder of this paper, we extend these results to a particular class of correlated states, which satisfy a so-called factorization property described below. As an application of this, we obtain bounds on the capacities of certain classes of classical-quantum channels with memory.\\\\Let $\cH$ be a finite dimensional Hilbert space, and define a family $\{\rho_n\}_{n\in\NN}$ of states such that for each $n$, $\rho_n\in\cD(\cH^{\otimes n})$. We say that the family $\{\rho_n\}_{n\in\NN}$ satisfies a \textit{(non-homogeneous) upper (lower) factorization property} if there exists an auxiliary family $\{\tilde{\rho}_k\}_{k\in\NN}$ of states on $\cD(\cH)$ and a constant $R>0$ such that for each $n\ge1$,
	\begin{align}
		&\rho_n\le R~ \rho_{n-1}\otimes \tilde{\rho}_n~~~~~~~~~~~\textit{upper factorization}\label{upfac}\\
		&\rho_n\ge R^{-1} \rho_{n-1}\otimes \tilde{\rho}_n~~~~~~~~~\textit{lower factorization}.\label{lfp}
	\end{align}
	The homogeneous case, where for each $n$, $\tilde{\rho}_n=\rho$, for some fixed state $\rho\in\cD(\cH)$ was first studied in \cite{HMO07,HMO08b}, where asymptotic results in the context of symmetric and asymmetric quantum hypothesis testing were derived for such families of states. Obviously, uncorrelated states satisfy lower and upper factorization, with $R=1$. Gibbs states of translation-invariant finite-range interactions were shown to satisfy both lower and upper homogeneous factorization properties for $R>1$ (see \cite{HMO07} Lemma 4.2). Similarly, finitely correlated states, as defined in \cite{FNW92}, were shown to satisfy a homogeneous upper factorization property as well, as a homogeneous lower factorization in some cases (see \cite{HMO07} Proposition 4.4 and Example 4.6). 
	The reason for introducing the non-homogeneous extension will become clear in \Cref{channelcoding}.\\\\
	An example of a family of states satisfying non-homogeneous upper factorization is provided by extending the definition in \cite{FNW92} of finitely correlated states as follows: suppose given a spin chain with one-site algebra $\mathcal{A}\subset \cB(\cH)$, and let	$\cB$ be a $C^*$-subalgebra of $\cB(\mathcal{K})$, for some finite dimensional Hilbert spaces $\mathcal{H}$ and $\mathcal{K}$, $\{\cE_n\}_{n\in\NN}$ a family of completely positive, unital (CPU) maps $\cE_n:\cA\otimes \cB\to \cB$, and $\rho$ a faithful state on $\cB$. Assume further that for any $n\in\NN$,
	\begin{align*}
		\tr_\cA\mathcal{E}_{n*}(\rho)=\rho,
		\end{align*} 
	where $\cE_{n*}:\cB\to\cA\otimes \cB$ stands for the pre-adjoint map of $\cE$. Construct then the family $\{\tau_n\}_{n\in\NN}$ of states on $\cA^{\otimes n}\otimes \cB$ as follows:
	\begin{align*}
		\tau_1:=\cE_{1*}(\rho),~~~~~~\tau_n:=(\id_\cA^{\otimes n-1}\otimes \cE_{n*})\circ~...~\circ(\id_\cA\otimes \cE_{2*})\circ\cE_{1*}(\rho),~~~~n=2,3,...
		\end{align*}
	To obtain a family of reduced states on the spin chain, we then trace out the auxiliary system $\cB$:
	\begin{align*}
		\rho_n:=\tr_\cB\tau_n.
	\end{align*}
In this case, we say that $(\cB,\{\cE_n\}_{n\in\NN},\rho)$ is a \textit{generating triple} for the family $\rho_n$.	In the special case in which $\mathcal{E}_n:=\cE$ for each $n$, where $\cE:\cA\otimes \cB\to \cB$ is a given CPU map, the states $\rho_n$ are the so-called finitely correlated states and provide a non-commutative generalization of the notion of (homogeneous) Markov chains. For the same reason, the above construction can be viewed as an extension of non-homogeneous Markov chains to the quantum setting. 
	\begin{proposition}\label{nhcsupp}
Non-homogeneous finitely correlated states satisfy the upper factorization property (\ref{upfac}), with $\tilde{\rho}_n:=\tr_\cB\cE_{n*}(\rho)$ and $R>1$.
	\end{proposition}	
	The proof of \Cref{nhcsupp} follows very closely the one for homogeneous finitely correlated states as given in Proposition 4.4 of \cite{HMO07}. For sake of completeness, we give a proof of it in \Cref{nhcsupp1}.\\\\
In a similar fashion as \cite{HMO07}, one can provide an example of non-homogeneous finitely correlated states satisfying the lower factorization property. To do so, assume that $\cB$ is a commutative algebra. Therefore it is isomorphic to the algebra $\mathcal{F}(\cX):\{f:\cX\to \CC\}$ of complex-valued functions on some finite set $\cX$. $\cF(\cX)$ is generated by the Dirac densities $\delta_x$, $x\in\cX$. Then, any CPTP map $\cE_*:\cB\to\cA\otimes \cB$ can be specified by its values on the functions $\delta_x$, and $\cE_*(\delta_x)$ can be uniquely decomposed in the form
\begin{align}\label{Estar}
	 \cE_{*}(\delta_x)=\sum_{y}~T_{xy}~\rho_{xy}\otimes \delta_y, 
	 \end{align}
	 where $\{T_{xy}:~y\in\cX\}$ forms a probability distribution on $\cX$ and $\rho_{xy}$ are states on $\cA$.
	 The following lemma was proved in \cite{HMO07} (see Example 4.6).
	 \begin{lemma}\label{lemmalfp}
	 	Let $\cE_*:\cB\to \cA\otimes \cB$ of the form given by \Cref{Estar}, and $\Phi:\cB\to \cB$, $\Phi(b)=\tr(\rho b)\mathbb{I}_\cB$. There exists $R>0$ such that $\cE_*-R(\cE_*\circ \Phi_*)$ is completely positive if and only if
	 	\begin{align}
	 		T>0~~~\text{ and }~~~ \supp \rho_{xy}=\supp\rho_{zy}~~~~\forall x,y,z\in\cX.\label{cond}
	 		\end{align} 
\end{lemma}
The following result is hence a direct consequence of \Cref{lemmalfp}:
\begin{proposition}
Let $\cB$ a commutative algebra with associated set $\cX$, and let $\{\rho_n\}_{n\in\NN}$ be a family of non-homogeneous finitely correlated states with generating triple $(\cB,\{\cE\}_{n\in\NN},\rho)$, where $\rho\equiv \sum_{x\in\cX} \rho_x\delta_x$. If for each $n$, $\cE_{n*}$ satisfies the conditions (\ref{cond}), then $\{\rho_n\}_{n\in\NN}$ satisfies the non-homogeneous lower factorization property (\ref{lfp}), with $\tilde{\rho}_n:=\sum_{x\in\cX}\rho_x\mathcal{E}_{n*}(\delta_x)$ and $R>0$.
\end{proposition}
The proof of this proposition follows exactly the same lines as the proof of \Cref{nhcsupp}, given in \Cref{nhcsupp1}, where the inequality comes from the fact that, by \Cref{lemmalfp}, there exists $R>0$ such that $\mathcal{E}_{n*}-R (\cE_{n*}\circ \Phi_*)$ is completely positive for each $n$. 
	
	\section{Finite blocklength hypothesis testing for correlated states}\label{finitecorrelated}
	In this section we derive finite blocklength bounds on the optimal type II errors in the case of sequences of correlated states. In order to do so, we use a variant of the proof of the Azuma-Hoeffding inequality, as well as the framework developed by Petz in \cite{Petz1986} where the monotonicity of the relative entropy is derived using the operator Jensen inequality \cite{D57,HP81}.
	Again we fix a sequence of finite dimensional Hilbert spaces $\{\cH^{\otimes n}\}_{n\in\NN}$ and let $\{\rho_n\}_{n\in\NN}$ and $\{\sigma_n\}_{n\in\NN}$ denote two sequences of states, where for each $n\in\NN$, $\rho_n,\sigma_n\in\cD(\cH^{\otimes n})$. Assume, moreover, that these sequences satisfy the (homogeneous) upper-factorization property: there exists $R>0$ such that
	\begin{align}\label{hufp}
		\rho_n\le R~\rho_{n-1}\otimes \rho_1,~~~~~\sigma_n\le R~\sigma_{n-1}\otimes \sigma_1,~~~~~n>1.
			\end{align}	
		\begin{theorem}\label{ahi} Given two sequences of states $\{\rho_n\}_{n\in\NN}$ and $\{\sigma_n\}_{n\in\NN}$ satisfying the upper factorization property (\ref{hufp}), with $R\ge 1$, the following bounds hold for any $0\le \eps\le 1$,
			\begin{align*}%\label{eq17}
				\beta_n(\eps)\le\left\{ \begin{aligned} &\e^{-nD(\rho_1\|\sigma_1)+c\sqrt{2n \log (R^n \eps^{-1}})} ~~~~~~\text{ if }\eps \ge R^n\e^{-nc^2/2},\\
					&\e^{-nD(\rho_1\|\sigma_1)+nc^2/2+\log(R^n\eps^{-1})}~~~~~~\text{ else,}
					\end{aligned}\right.
				\end{align*}
				where $c:=\|\log\Delta_{\sigma_1\|\rho_1}+D(\rho_1\|\sigma_1)\id\|$. Similarly,
					\begin{align*}
						\tilde{\beta}_n(r)\le\left\{ \begin{aligned} &
							\e^{-nD(\rho_1\|\sigma_1)+nc\sqrt{2(r+\log R)}} ~~~~~~~\text{ if } r\le c^2/2-\log R,\\
							&\e^{-n(D(\rho_1\|\sigma_1)-c^2/2-r-\log R)}~~~~~~~\text{ else}.
						\end{aligned}\right.
					\end{align*}
		
		\end{theorem}
		\begin{proof}
		As before, $\Omega_n:=\rho_n^{1/2}$, $n\in\NN$. Now for each $n\in\NN$, define the map $V_{n-1}:\cB(\cH^{\otimes n})(\Omega_{\rho_{n-1}}\otimes\Omega_{\rho_1})\to \cB(\cH^{\otimes n}) \Omega_{\rho_n}$ by
			\begin{align}\label{vn}
				V_{n-1}(X_n(\Omega_{\rho_{n-1}}\otimes \Omega_{\rho_1}))=R^{-1/2}X_n\Omega_{\rho_n}, ~~~~~~~~\forall X_n\in\cB(\cH^{\otimes n}).
			\end{align}
			In particular, for $X_n:=\mathbb{I}_n$, we have
			\begin{align}\label{eq23}
				V_{n-1}(\Omega_{n-1}\otimes \Omega_{\rho_1})=R^{-1/2}\Omega_{\rho_n}.
				\end{align}
			One can verify that $V_{n-1}$ is a contraction: for any $X_n \in \cB(\cH^{\otimes n})$,
			\begin{align}
				\langle V_{n-1}(X_n(\Omega_{\rho_{n-1}}\otimes\Omega_{\rho_1})),V_{n-1}(X_n(\Omega_{\rho_{n-1}}\otimes\Omega_{\rho_1}))\rangle&=R^{-1}\langle X_n\Omega_{\rho_{n}},X_n\Omega_{\rho_{n}}\rangle\nonumber\\
				&=R^{-1}\tr(X_n\rho_{n}X_n^*)\nonumber\\
			&	\le \tr(X_n(\rho_{n-1}\otimes \rho_1) X_n^{*})\label{ineq1}\\
				&=\langle X_n(\Omega_{\rho_{n-1}}\otimes\Omega_{\rho_1}),X_n(\Omega_{\rho_{n-1}}\otimes\Omega_{\rho_1})\rangle\nonumber,
				\end{align}
				where the inequality (\ref{ineq1}) follows from the upper factorization property (\ref{hufp}).
				Moreover, for any $X_n\in\cB(\cH^{\otimes n})$,
			\begin{align}
	\langle V_{n-1} (X_n (\Omega_{\rho_{n-1}}\otimes \Omega_{\rho_1}), &~\Delta_{\sigma_n|\rho_n} V_{n-1}(\Omega_{\rho_{n-1}}\otimes \Omega_{\rho_1})\rangle
	=R^{-1}\langle X_n\Omega_{\rho_n},\Delta_{\sigma_n|\rho_n}X_n\Omega_{\rho_n}\rangle\nonumber\\
&				=R^{-1}\tr (\sigma_n X_nX_n^*)\nonumber\\
			&	\le  \tr((\sigma_{n-1}\otimes \sigma_1)X_nX_n^*)\label{ineq2}\\
				&=\langle X_n(\Omega_{\rho_{n-1}}\otimes \Omega_{\rho_1}),\Delta_{\sigma_{n-1}\otimes \sigma_1|\rho_{n-1}\otimes \rho_1}(X_n(\Omega_{\rho_{n-1}}\otimes \Omega_{\rho_1}))\rangle,&\nonumber
			\end{align}
			where again, the inequality (\ref{ineq2}) follows from the upper factorization property (\ref{hufp}). Hence, on $\cB(\cH^{\otimes n})(\Omega_{\rho_{n-1}}\otimes \Omega_{\rho_1})$,
			\begin{align}\label{eq4bis}
				V_{n-1}^*\Delta_{\sigma_n|\rho_n}V_{n-1}\le \Delta_{\sigma_{n-1}\otimes \sigma_1|\rho_{n-1}\otimes \rho_1}=\Delta_{\sigma_{n-1}|\rho_{n-1}}\otimes \Delta_{\sigma_1|\rho_1}.
				\end{align}
				Fix $0\le t\le 1$ and $\lambda\in\RR$. By functional calculus, we obtain the following Markov-like inequality:
				\begin{align}\label{eq10}
					P_{[\lambda,\infty)}(\log(\Delta_{\sigma_n|\rho_n}))=P_{[\e^{\lambda t} ,\infty)}(\Delta_{\sigma_n|\rho_n}^t)\le \e^{-\lambda t}\Delta_{\sigma_n|\rho_n}^t.
					\end{align}
				Further, using \reff{hufp},
				\begin{align}
					\langle \Omega_{\rho_n},\Delta_{\sigma_n|\rho_n}^t(\Omega_{\rho_n})\rangle&=R~\langle V_{n-1}(\Omega_{\rho_{n-1}}\otimes \Omega_{\rho_1}), \Delta_{\sigma_n|\rho_n}^tV_{n-1}(\Omega_{\rho_{n-1}}\otimes\Omega_{\rho_1})\rangle\nonumber\\
					&=R~\langle (\Omega_{\rho_{n-1}}\otimes \Omega_{\rho_1}),V_{n-1}^* \Delta_{\sigma_n|\rho_n}^tV_{n-1}(\Omega_{\rho_{n-1}}\otimes\Omega_{\rho_1})\rangle\nonumber\\
					&\le R~\langle (\Omega_{\rho_{n-1}}\otimes \Omega_{\rho_1}),(V_{n-1}^* \Delta_{\sigma_n|\rho_n}V_{n-1})^t(\Omega_{\rho_{n-1}}\otimes\Omega_{\rho_1})\rangle\nonumber\\
					&\le R~\langle (\Omega_{\rho_{n-1}}\otimes \Omega_{\rho_1}),(\Delta_{\sigma_{n-1}|\rho_{n-1}}\otimes \Delta_{\sigma_1|\rho_1})^t(\Omega_{\rho_{n-1}}\otimes\Omega_{\rho_1})\rangle\nonumber\\
					&=R~\langle \Omega_{\rho_{n-1}}, \Delta_{\sigma_{n-1}|\rho_{n-1}}^t(\Omega_{\sigma_{n-1}})\rangle\langle \Omega_{\rho_1}, \Delta_{\sigma_1|\rho_1}^t(\Omega_{\rho_1})\rangle,\label{eq24}
				\end{align}	
				where the third line follows from \Cref{Jensen} and the operator concavity of $x\mapsto x^t$, and the fourth line follows from the operator monotonicity of $x\mapsto x^t$ for $t\in[0,1]$ as well as \reff{eq4bis}. By iterating \reff{eq24} $n-1$ times, we obtain:
				\begin{align}\label{eq11}
					\langle \Omega_{\rho_n},\Delta_{\sigma_n|\rho_n}^t(\Omega_{\rho_n})\rangle &\le R^n \langle \Omega_{\rho_1}, \Delta_{\sigma_1|\rho_1}^t(\Omega_{\rho_1})\rangle^n=R^n~
					\langle \Omega_{\rho_1},\e^{t(\log(\Delta_{\sigma_1|\rho_1}))}(\Omega_{\rho_1})\rangle^n\nonumber\\
					& 	=R^n\e^{-ntD(\rho_1\|\sigma_1)}~\langle \Omega_{\rho_1},\e^{t(\log(\Delta_{\sigma_1|\rho_1})+D(\rho_1\|\sigma_1))}(\Omega_{\rho_1})\rangle^n.
					\end{align}
				Now due to the convexity of the exponential function, we directly get that:
					\begin{align*}%\label{eq12}
						\e^{ts}\le \frac{1}{2c}\left( \e^{tc}-\e^{-tc}\right)s+\frac{1}{2}\left(\e^{tc}+\e^{-tc}\right)~~~~~~~~-c\le s\le c.
					\end{align*}
					Hence, we obtain the following by functional calculus: for
					\begin{align}\label{c} -c~\id \le\log\Delta_{\sigma_1|\rho_1}+D(\rho_1\|\sigma_1)~\id\le c~\id, ~~~\text{where} ~~~
					c \equiv\|\log\Delta_{\sigma_1|\rho_1}+D(\rho_1\|\sigma_1)~\id\|,
					\end{align}
					\begin{align*}
						\langle \Omega_{\rho_1},&~\Delta_{\sigma_1|\rho_1}^t(\Omega_{\rho_1})\rangle\\
						&\le\e^{-tD(\rho_1\|\sigma_1)} \left[\frac{1}{2c}\left(\e^{tc}-\e^{-tc}\right)(\langle \Omega_{\rho_1}, \log \Delta_{\sigma_1|\rho_1}(\Omega_{\rho_1})\rangle+D(\rho_1\|\sigma_1))+\frac{1}{2}\left( \e^{tc}+\e^{-tc}\right)\right]\nonumber\\
						&= \e^{-tD(\rho_1\|\sigma_1)}\frac{1}{2}\left(\e^{tc}+\e^{-tc}\right)\le \e^{-tD(\rho_1\|\sigma_1)+t^2c^2/2},
						\end{align*}
						where, in the last line, we used the fact that $ \langle \Omega_{\rho_1}, \log \Delta_{\sigma_1|\rho_1}(\Omega_{\rho_1})\rangle=-D(\rho_1\|\sigma_1)$, and the inequality \[\frac{1}{2}(\e^{u}+\e^{-u})\le \e^{u^2/2}\] which can be verified by Taylor expansion.
						The above bound, together with \reff{eq10} and \reff{eq11}, yields,
						\begin{align*}
							\langle \Omega_{\rho_n}, P_{[\lambda,\infty)}(\log\Delta_{\sigma_n|\rho_n})(\Omega_{\rho_n})\rangle\le R^n \e^{-(\lambda+nD(\rho_1\|\sigma_1)) t}\e^{nt^2c^2/2},
							\end{align*}
							where $c$ is given in \reff{c}. 
						Optimizing over $t$, we get
	\begin{align}\label{up}
		\langle \Omega_{\rho_n}, P_{[\lambda,\infty)}(\log\Delta_{\sigma_n|\rho_n})(\Omega_{\rho_n})\rangle\le\left\{
		\begin{aligned}
		&R^n \e^{-(\lambda+nD(\rho_1\|\sigma_1))^2/(2nc^2)}~~~ \text{ if } \lambda\le n(c^2-D(\rho_1\|\sigma_1)) ~~(i) \\
		&R^n \e^{-(\lambda+nD(\rho_1\|\sigma_1))+nc^2/2}~~~~\text{ else~~~~~~~~~~~~~~~~~~~~~~~~~~~~}(ii).
		 \end{aligned}
		 \right.
	\end{align}
	The reason for this separation of cases comes from the fact that the optimization procedure can only be carried out for $0\le t\le 1$. The condition in (i) is found so that the global optimizer $t$ satisfies this condition. On the contrary, (ii) corresponds to the case in which the global minimizer found is greater than $1$: in this case we take the best minimizer within the interval $[0,1]$, which is $t=1$.
			Now given a sequence $\{L_n\}_{n\in\NN}$ of positive numbers, and any $n\in\NN$, choose $\lambda=-\log L_n$. We recall that by \Cref{prop_1}, there exists a sequence of tests $\{T_n\}_{n\in\NN}$  such that $\beta(T_n)\le L_n^{-1}$ and
						\begin{align*}
							\alpha(T_n) \le\langle \Omega_{\rho_n}, P_{(-\log L_n,\infty)} (\log \Delta_{\sigma_n|\rho_n})(\Omega_{\rho_n})\rangle.
						\end{align*}						
	Setting $\eps$ to be the right hand side of \reff{up}(i), we end up with \[\log \beta_{n}(\eps)\le-\log L_n= -nD(\rho_1\|\sigma_1)+c\sqrt{2n\log(R^n\eps^{-1})},\]
						which satisfies the condition $-\log L_n\le n(c^2-D(\rho_1\|\sigma_1))$ for $\eps\ge R^n\e^{-nc^2/2}$. We here implicitly used that $R\ge 1$ so that the quantity inside of the square root is indeed positive. If instead we set $\eps$ to be the right hand side of \reff{up}(ii), we end up with
						\[\log \beta_{n}(\eps)\le-\log L_n= -n(D(\rho_1\|\sigma_1)-\log R)+nc^2/2 +\log (\eps^{-1}),\]
						this bound being achieved for $-\log L_n\ge n(c^2-D(\rho_1\|\sigma_1))$, i.e. for $\eps \le R^n \e^{-n c^2/2}$. The Hoeffding-type bounds of parameter $r$ are derived similarly by setting the upper bounds in \reff{up} equal to $\e^{-nr}$.
			\qed
		\end{proof}	
		The proof of \Cref{ahi} can be readily extended to the case of sequences $\{\rho_n\}_{n\in\NN}$ and $\{\sigma_n\}_{n\in\NN}$ of states satisfying the non-homogeneous upper factorization property:
		\begin{corollary}\label{cor2}
			Consider two sequences of states $\{\rho_n\}_{n\in\NN}$ and $\{\sigma_n\}_{n\in\NN}$ satisfying the non-homogeneous upper factorization property (\ref{upfac}) with auxiliary families of states $\{\tilde{\rho}_n\}_{n\in\NN}$ and $\{\tilde{\sigma}_n\}_{n\in\NN}$, and $R\ge 1$. Then, for any $0\le\eps\le 1$ the following bounds hold for the Stein error:
				\begin{align}\label{eq17}
					\beta_n(\eps)\le\left\{ \begin{aligned} &\e^{-D_n(\{\rho_n\}\|\{\sigma_n\})+C_n\sqrt{2 \log (R^n \eps^{-1}})} ~~~~~~~\text{ if }\eps \ge R^n\e^{-C_n^2/2},\\
						&\e^{-D_n(\{\rho_n\}\|\{\sigma_n\})+C_n^2/2+\log(R^n\eps^{-1})}~~~~~~~\text{ else,}
					\end{aligned}\right.
				\end{align}
				where $C_n^2:=\sum_{k=1}^n c_k^2$, with $c_k:=\|\log\Delta_{\tilde{\sigma}_k\|\tilde{\rho}_k}+D(\tilde{\rho}_k\|\tilde{\sigma}_k)\id\|$, and
				\begin{align*}
					D_n(\{\rho_n\}\|\{\sigma_n\}):=\sum_{k=1}^n D(\tilde{\rho}_k\|\tilde{\sigma}_k).
				\end{align*}
					 Similarly, for any $r>0$ the following bounds hold for the Hoeffding error:
				\begin{align*}
					\tilde{\beta}_n(r)\le\left\{ \begin{aligned} &
						\e^{-D_n(\{\rho_n\}\|\{\sigma_n\})+C_n\sqrt{2n(r+\log R)}} ~~~~~~~~~~\text{ if } r\le C_n^2/(2n)-\log R,\\
						&\e^{-D_n(\{\rho_n\}\|\{\sigma_n\})+C_n^2/2+nr+n\log R}~~~~~~~~\text{ else}.
					\end{aligned}\right.
				\end{align*}
		\end{corollary}

\section{Moderate deviation hypothesis testing for correlated states}\label{moderatedev}
In this section, we derive bounds on the asymptotic behavior of the optimal type II error in the moderate deviation regime, which interpolates between the regime of large deviations (Stein's lemma) and the one of the Central Limit Theorem (second order asymptotics).\\\\A \textit{moderate sequence} of real numbers $\{a_n\}_{n\in\NN}$ has as a defining property that $a_n\to0$ and $\sqrt{n}a_n\to\infty$, as $n\to\infty$. A typical example of such a sequence is given by the choice $a_n=n^{-t}$ for some $t\in(0,1/2)$. Recently the moderate deviation analysis of quantum hypothesis testing for uncorrelated states was studied by Cheng et al.~\cite{HC17}, and Chubb et al.~\cite{CTT17}. Our results extend theirs to families of correlated (i.e.~non i.i.d.) states satisfying the upper and/or lower factorization property. We first restate the result of \cite{CTT17} for sake of completeness. For two states $\rho,\sigma\in\cD(\cH)$, consider the $\eps$\textit{-hypothesis testing relative entropy}\footnote{This quantity is also referred to as $\eps$-hypothesis testing divergence, e.g. in \cite{CTT17}.}, introduced by Wang and Renner in \cite{WR12}:
\begin{align}\label{DH}
	D_H^\eps(\rho\|\sigma):=-\log\beta(\eps),
	\end{align}
	where $\beta(\eps)$ is the one-shot optimal type II error in the hypothesis testing problem with null hypothesis $\rho$ and alternative hypothesis $\sigma$:
	\begin{align*}
		\beta(\eps):=\min_{0\le T\le \mathbb{I}}\{\beta(T)|~\alpha(T)\le \eps\}.
		\end{align*}
\begin{theorem}[\cite{CTT17} Theorem 1]\label{theoremCTT17}
For $\rho,\sigma>0$, a moderate sequence $\{a_n\}_{n\in\NN}$ and $\eps_n=\e^{-na_n^2}$, the $\eps_n$-hypothesis testing divergences scales as
\begin{align*}
	\frac{1}{n}D_H^{\eps_n}(\rho^{\otimes n}\|\sigma^{\otimes n})= D(\rho\|\sigma)-\sqrt{2V(\rho\|\sigma)} a_n+\circ(a_n),
	\end{align*}
	where $D(\rho\|\sigma)$ is the quantum relative entropy defined in \Cref{qrelent}, and $V(\rho\|\sigma)$ is the quantum information variance defined in \Cref{V}.
	\end{theorem}

The proof of \Cref{theoremCTT17} relies on a reduction of the problem to the one of bounding the law of a sum of independent random variables (which is a typical example of a martingale) from above and below. Such a reduction does not directly carry over to the case of correlated states, and hence a finer analysis of the error probabilities in the moderate deviation regime is needed.\\\\
 We recall that a sequence $\{\rho_n\}_{n\in\NN}$ of states $\rho_n\in\cD(\cH^{\otimes n})$ satisfies the (homogeneous) upper (lower)-factorization property if there exists $R>0$ such that for each $n>1$:
\begin{align}
	&	\rho_n\le R~\rho_{n-1}\otimes \rho_1~~~~~~~\text{(upper factorization)}\label{ufpp},\\
&	\rho_n\ge R^{-1}\rho_{n-1}\otimes \rho_1~~~~~\text{(lower factorization)}.\label{hlfp}
	\end{align} 
\begin{theorem}\label{moderate}
	Let $\{\rho_n\}_{n\in\NN}$ and $\{\sigma_n\}_{n\in\NN}$ be two sequences of states $\rho_n,\sigma_n\in\cD(\cH^{\otimes n})$, and, for each $n$, define $\eps_n:=\e^{-na_n^2}$, where $a_n$ is a given moderate sequence. If $\{\rho_n\}_{n\in\NN}$ and $\{\sigma_n\}_{n\in\NN}$ satisfy the upper factorization property (\ref{ufpp}), with $$0\le \log R<{(4-\e^c)(\e^c-1)^2 V(\rho_1\|\sigma_1)/(6c^2)},$$ where $c:=\|\log\Delta_{\sigma_1\|\rho_1}+D(\rho_1\|\sigma_1)\id\|_\infty <\log 4$. Then, for all $n$ large enough:
			\begin{align}\label{inward}
				\frac{1}{n}D_H^{\eps_n}(\rho_n\|\sigma_n)\ge D(\rho_1\|\sigma_1)-\sqrt{2V(\rho_1\|\sigma_1)}\left(\frac{3}{4-e^c}(\log R+a_n^2)\right)^{1/2}
			\end{align}	
		If $\{\rho_n\}_{n\in\NN}$ and $\{\sigma_n\}_{n\in\NN}$ satisfy the lower factorization property (\ref{hlfp}) with $R>1$, then 
				\begin{align}\label{out}
				\frac{1}{n}	D_H^{\eps_n}(\rho_n\|\sigma_n)\le D(\rho_1\|\sigma_1)-\sqrt{2V(\rho_1\|\sigma_1)}a_n+\circ( a_n).
				\end{align}
				
	\end{theorem}

\begin{remark} As mentioned previously, Gibbs states satisfy both the upper and lower factorization property with $R>1$. Hence for these states, both the upper and lower bounds of \Cref{moderate} hold (see \Cref{corr}). Finitely correlated state are known to satisfy the upper factorization property for $R>1$.
\end{remark}
	\begin{proof}
		The proof of \Cref{inward} closely follows the one of \Cref{ahi}. We start by recalling \reff{eq11}: let $n\in\NN$, then for any $0\le t\le 1$,
			\begin{align}\label{eq20}
				\langle \Omega_{\rho_n},\Delta_{\sigma_n|\rho_n}^t(\Omega_{\rho_n})\rangle &\le 	R^n\e^{-ntD(\rho_1\|\sigma_1)}~\langle \Omega_{\rho_1},\e^{t(\log \Delta_{\sigma_1|\rho_1}+D(\rho_1\|\sigma_1))}(\Omega_{\rho_1})\rangle^n.
			\end{align}
Define $X$ to be the classical centered random variable associated with $\log\Delta_{\sigma_1|\rho_1}+D(\rho_1\|\sigma_1)$ and $\Omega_{\rho_1}$ as given through \Cref{spec-meas}. Then \reff{eq20} can be rewritten as
\begin{align}\label{eq25}
		\langle \Omega_{\rho_n},\Delta_{\sigma_n|\rho_n}^t(\Omega_{\rho_n})\rangle \le R^n\e^{-ntD(\rho_1\|\sigma_1)}\e^{\Psi_n(t)},
		\end{align}
		where $\Psi_n(t):=n\log \mathbb{E}[\e^{tX}]$ is the cumulant generating function of the sum of n i.i.d.~random variables $X_i$ of law identical to the one of $X$. Now, by  \Cref{prop_1}, for any sequence $\{L_n\}_{n\in\NN}$ of positive numbers such that, for each $n$, $\log L_n\le nD(\rho_1\|\sigma_1)$, there exists a sequence $\{T_n\}_{n\in\NN}$ of tests such that for each $n$, $\beta(T_n)\le L_n^{-1}$ and 
\begin{align}
	\alpha(T_n)&\le \langle \Omega_{\rho_n}, P_{(-\log L_n, \infty)}(\log\Delta_{\sigma_n|\rho_n})(\Omega_{\rho_n})\rangle\nonumber\\
	&=\langle \Omega_{\rho_n}, P_{( L_n^{- t}, \infty)}(\Delta_{\sigma_n|\rho_n}^t)(\Omega_{\rho_n})\rangle\nonumber\\
	&\le L_n^t \langle \Omega_{\rho_n}, \Delta_{\sigma_n|\rho_n}^t(\Omega_{\rho_n})\rangle\nonumber\\
	&\le R^n\e^{-[t(nD(\rho_1\|\sigma_1)-\log L_n)-\Psi_n(t)]}.\label{eq26}
\end{align}
where the inequality in the third line comes from the Markov-type inequality \[P_{(-\log L_n, \infty)}(\log\Delta_{\sigma_n|\rho_n})\le L_n^t\Delta_{\sigma_n|\rho_n}^t,\] and the last inequality in \reff{eq26} comes from (\ref{eq25}). Following a similar Cram\'{e}r-Chernoff method as the one used in the proof of Bennett's inequality (see e.g. Theorem 2.9 of \cite{BLM13}), one can optimize the bound in \reff{eq26} over $t$ to obtain the following upper bound on $\alpha(T_n)$:
\begin{align}\label{eq21}
	\alpha(T_n)\le R^n \exp\left[{-nV(\rho_1\|\sigma_1)h\left(c\frac{nD(\rho_1\|\sigma_1)-\log L_n}{nV(\rho_1\|\sigma_1)}\right)}/c^2\right],
\end{align}	
with $h(u):=(1+u)\log (1+u)-u$, for $u>0$, and where $c:=\|\log\Delta_{\sigma_1\|\rho_1}+D(\rho_1\|\sigma_1)\id\|_\infty$. Note that the argument of the function $h$ in \reff{eq21} is indeed positive since we chose $\log L_n\le nD(\rho_1\|\sigma_1)$. The optimizer, $t_0$, is given by \[t_0:=\frac{1}{c}\log\left(1+c\frac{nD(\rho_1\|\sigma_1)-\log L_n}{nV(\rho_1\|\sigma_1)}\right).\] Imposing $t_0\le 1$, we find that $L_n$ has to satisfy
\begin{align}\label{eq2}
	0\le nD(\rho_1\|\sigma_1)-\log L_n\le (e^c-1)nV(\rho_1\|\sigma_1)/c.
	\end{align}
	Using now $h(u)\ge u^2/2-u^3/6$ for $u\ge0$ (see \cite{S12} Lemma 1), \reff{eq21} yields
	\begin{align*}
		\log \alpha(T_n)&\le n\log R-\frac{nV(\rho_1\|\sigma_1)}{c^2}\left[\frac{1}{2} \left(c\frac{nD(\rho_1\|\sigma_1)-\log L_n}{nV(\rho_1\|\sigma_1)}\right)^2-\frac{1}{6}\left(c\frac{nD(\rho_1\|\sigma_1)-\log L_n}{nV(\rho_1\|\sigma_1)}\right)^3\right]\\
		&\le n\log R-\frac{(nD(\rho_1\|\sigma_1)-\log L_n)^2}{2nV(\rho_1\|\sigma_1)}\left[(4-\e^c)/3\right]\\
		&\equiv\log \eps_n= -na_n^2,
\end{align*}	
As in the proof of \Cref{ahi}, the assumption that $R\ge 1$ is crucial in order for the last equation to have a solution. From the above we obtain the following expression for $\log L_n$ in terms of $a_n$:
\begin{align*}
\log L_n=	nD(\rho_1\|\sigma_1)-n\sqrt{2V(\rho_1\|\sigma_1)}\left(\frac{3}{4-e^c}(\log R+a_n^2)\right)^{1/2},
\end{align*}	
provided that \reff{eq2} is satisfied for large $n$. This condition imposes R to be smaller than $\exp({(4-\e^c)(\e^c-1)^2 V(\rho_1\|\sigma_1)/(6c^2)})$. From the above expression and the inequality $\beta(T_n)\le L_n^{-1}$ we get:
\begin{align*}
	\frac{1}{n}D_H^{\eps_n}(\rho_n\|\sigma_n)\ge D(\rho_1\|\sigma_1)-\sqrt{2V(\rho_1\|\sigma_1)}\left(\frac{3}{4-e^c}(\log R+a_n^2)\right)^{1/2}
\end{align*}	
	Next we prove \reff{out}. We first recall the following result from \cite{Jaksicetal} which we already used in \cite{DPR16} in the context of second order asymptotic analysis,
			\begin{lemma}\label{boundsD}
				For any $\theta_n,v_n\in\RR$:
				\begin{align}\label{ineq}
					\frac{\e^{-\theta_n}}{1+\e^{v_n-\theta_n}}\langle \Omega_{\rho_n},P_{[-v_n,\infty)}(\log\Delta_{\sigma_n|\rho_n})(\Omega_{\rho_n})\rangle\le \operatorname{e_{sym}^*}(\sigma_n,\e^{-\theta_n}\rho_n), 
				\end{align}
				where $\operatorname{e_{sym}^*}(A,B):=\inf_{0\le T\le \mathbb{I}} \{\tr(A(\mathbb{I}-T))+\tr(BT)\}$ is the minimum total probability of error in symmetric hypothesis testing, the definition here being extended to unnormalized operators $A,B>0$.
			\end{lemma}
Fix $\theta_n,v_n\in\RR$ to be specified later, then \reff{ineq} implies that, for any test $0\le T_n\le \mathbb{I}_n$ for which $\alpha(T_n)\le \eps_n$, we have
\begin{align}
	\beta(T_n)\ge \e^{-\theta_n}\left(\frac{\langle \Omega_{\rho_n}, P_{[-v_n,\infty)}(\log\Delta_{\sigma_n|\rho_n})(\Omega_{\rho_n})\rangle}{1+\e^{v_n-\theta_n}}-\eps_n\right)\label{eq13}
			\end{align}
			Now for $0\le \lambda\le 1$, 
			\begin{align}
				\langle\Omega_{\rho_n}, P_{[-v_n,\infty)}(\log\Delta_{\sigma_n|\rho_n})(\Omega_{\rho_n})\rangle&=				\langle\Omega_{\rho_n}, P_{[\e^{-\lambda v_n},\infty)}(\Delta_{\sigma_n|\rho_n}^{\lambda})(\Omega_{\rho_n})\rangle\nonumber\\
				&\ge 1-\e^{-\lambda v_n}\langle \Omega_{\rho_n}, \Delta_{\sigma_n|\rho_n}^{-\lambda}(\Omega_{\rho_n})\rangle,\label{markov}
			\end{align}
			where the last line above follows from the reverse Markov inequality (see \Cref{reversemarkov}).\\\\
			We now obtain a bound on $\langle \Omega_{\rho_n}, \Delta_{\sigma_n|\rho_n}^{-\lambda}(\Omega_{\rho_n})\rangle$: consider the map \[W_n:\cB(\cH^{\otimes n})\Omega_{\rho_n}\to \cB(\cH^{\otimes n})(\Omega_{\rho_{n-1}}\otimes \Omega_{\rho_1}),\] where
			\[W_{n-1}(X_n\Omega_{\rho_n})=R^{-1/2}X_n(\Omega_{\rho_{n-1}}\otimes\Omega_{\rho_1}).\]
			Similarly to the map $V_{n-1}$ defined in \Cref{vn}, one can show that the map $W_{n-1}$ is a contraction, by the lower factorization property: for any $X_n\in\cB(\cH^{\otimes n})$,
					\begin{align}
						\langle W_{n-1}(X_n\Omega_{\rho_{n}}),W_{n-1}(X_n \Omega_{\rho_{n}})\rangle&=R^{-1}\langle X_n(\Omega_{\rho_{n-1}}\otimes\Omega_{\rho_1}),X_n(\Omega_{\rho_{n-1}}\otimes \Omega_{\rho_1})\rangle\nonumber\\
						&=R^{-1}\tr(X_n(\rho_{n-1}\otimes \rho_1)X_n^*)\nonumber\\
					&	\le \tr(X_n\rho_{n} X_n^{*})\nonumber\\
					&	=\langle X_n\Omega_{\rho_{n}}
						,X_n\Omega_{\rho_{n}})\rangle.\nonumber
					\end{align}
					Moreover, 
								\begin{align*}
									\langle W_{n-1} (X_n \Omega_{\rho_{n}}),~& \Delta_{\sigma_{n-1}\otimes \sigma_1|\rho_{n-1}\otimes\rho_1} W_{n-1}(X_n\Omega_{\rho_{n}})\rangle\\
									&=R^{-1}\langle X_n(\Omega_{\rho_{n-1}}\otimes \Omega_{\rho_1}),\Delta_{\sigma_{n-1}\otimes\sigma_1|\rho_{n-1}\otimes\rho_1}X_n(\Omega_{\rho_{n-1}}\otimes \Omega_{\rho_1})\rangle\\
								&	=R^{-1}\tr ((\sigma_{n-1}\otimes\sigma_1) X_nX_n^*)\le  \tr(\sigma_{n}X_nX_n^*)\\
								&	=\langle X_n\Omega_{\rho_{n}},\Delta_{\sigma_{n} |\rho_{n}}(X_n\Omega_{\rho_{n}})\rangle,
								\end{align*}
								so that
								\begin{align}\label{eq4}
									W_{n-1}^*\Delta_{\sigma_{n-1}\otimes \sigma_1|\rho_{n-1}\otimes \rho_1}W_{n-1}\le \Delta_{\sigma_{n}|\rho_{n}}.
								\end{align}
								This implies that for $0\le\lambda\le 1$,
								\begin{align}
									\langle \Omega_{\rho_n},\Delta_{\sigma_n|\rho_n}^{-\lambda}(\Omega_{\rho_n})\rangle&\le \langle\Omega_{\rho_n}, (W_{n-1}^*(\Delta_{\sigma_{n-1}\otimes\sigma_1|\rho_{n-1}\otimes \rho_1})W_{n-1})^{-\lambda}(\Omega_{\rho_n})\rangle\nonumber\\
									&\le \langle \Omega_{\rho_n}, W_{n-1}^*(\Delta_{\sigma_{n-1}\otimes \sigma_1|\rho_{n-1}\otimes\rho_1})^{-\lambda}W_{n-1}\Omega_{\rho_n}\rangle\nonumber\\
&=								R^{-1} \langle \Omega_{\rho_{n-1}}, \Delta^{-\lambda}_{\sigma_{n-1}|\rho_{n-1}}(\Omega_{\rho_{n-1}})\rangle\langle\Omega_{\rho_1}, \Delta_{\sigma_1|\rho_1}^{-\lambda}(\Omega_{\rho_1})\rangle\nonumber\\
&\le R^{-n} \langle \Omega_{\rho_1},\Delta_{\sigma_1|\rho_1}^{-\lambda}(\Omega_{\rho_1})\rangle^n,\label{ineq3}
									\end{align}
									where the first inequality follows from operator monotonicity of $x\mapsto -x^{-\lambda}$ as well as \reff{eq4}, the second one from operator convexity of $x\mapsto x^{-\lambda}$, and the last one by iterating the process $n-1$ times. Now by using inequalities (\ref{markov}) and (\ref{ineq3}), inequality (\ref{eq13}) implies
									\begin{align}\label{eq28}
										\beta(T_n)\ge \e^{-\theta_n}\left(\frac{1-\e^{-\lambda v_n}R^{-n}\langle \Omega_{\rho_1}, \Delta^{-\lambda}_{\sigma_1|\rho_1}(\Omega_{\rho_1})\rangle^n}{1+\e^{v_n-\theta_n}}-\eps_n\right).
										\end{align}
										Defining $\xi\equiv \log\Delta_{\sigma_1|\rho_1}+D(\rho_1\|\sigma_1)\id$, 
										\begin{align*}
											\langle \Omega_{\rho_1}, \Delta_{\sigma_1|\rho_1}^{-\lambda}(\Omega_{\rho_1})\rangle^n=		\langle \Omega_{\rho_1^{\otimes n}}, \left(\e^{-\lambda \xi}\right)^{\otimes n}(\Omega_{\rho_1^{\otimes n}})\rangle\e^{n\lambda D(\rho_1\|\sigma_1)}
											\end{align*}
											which, by functional calculus, can be interpreted as the moment generating function of a sum of centered i.i.d.~random variables. Following the steps of the proof of the Berry-Esseen theorem \cite{D10}, one can find how quickly its associated cumulant generating function gets close to the one of a Gaussian random variable. More precisely, 
											\begin{align*}
											\langle \Omega_{\rho_1},\e^{-\lambda\xi}(\Omega_{\rho_1})\rangle&= \langle \Omega_{\rho_1}, (\id-\lambda\xi+\frac{(\lambda \xi)^2}{2}+\mathcal{O}\left({\lambda^3 \xi^3}\right))(\Omega_{\rho_1})\rangle\\
											&=1  +\frac{\lambda^2}{2}V(\rho_1\|\sigma_1)+\mathcal{O}(\lambda^3\langle \Omega_{\rho_1}, \xi^3(\Omega_{\rho_1})\rangle)\\
											&=\e^{\frac{\lambda^2}{2}V(\rho_1\|\sigma_1)+\mathcal{O}(\lambda^3\langle \Omega_{\rho_1},\xi^3(\Omega_{\rho_1})\rangle)},
													\end{align*}
					where we used the fact that $\langle \Omega_{\rho_1}, \xi(\Omega_{\rho_1})\rangle=0$ in the second line. Taking $\lambda=1/\sqrt{n}$, \reff{eq28} reduces to
											\begin{align*}
															\beta(T_n)\ge \e^{-\theta_n}\left(\frac{1-\e^{- v_n/\sqrt{n}}\e^{  \sqrt{n}D(\rho_1\|\sigma_1)}R^{-n}\e^{\frac{1}{2}V(\rho_1\|\sigma_1)+\mathcal{O}(\langle \Omega_{\rho_1},\xi^3(\Omega_{\rho_1})\rangle/\sqrt{n})}}{1+\e^{v_n-\theta_n}}-\eps_n\right).
															\end{align*}
											Choosing then, for any $\eta>0$,
											\begin{align*}
												\theta_n:=nD(\rho_1\|\sigma_1)-n\sqrt{2V(\rho_1\|\sigma_1)}a_n+\eta na_n/2,
											\end{align*}
											and $v_n:=\theta_n-\eta na_n/2 $, the term between parentheses is bounded, so that, for $n$ large enough,
					\begin{align*}
				\frac{1}{n}D_H^{\eps_n}(\rho_1\|\sigma_1)\le\frac{1}{n}\theta_n+\eta a_n/2=D(\rho_1\|\sigma_1)-\sqrt{2V(\rho_1\|\sigma_1)}a_n+\eta a_n.
				\end{align*}

		\qed
		\end{proof}
		\begin{remark}
					\reff{inward} is found similarly to the bounds derived in \Cref{ahi}, and the proof is inspired by the proof of Bernstein's concentration inequality (see \cite{JZ13,RS13}). As in the case of the Azuma-Hoeffding-type bound (\ref{eq17}), the difference with the derivation of Bernstein's inequality for tracial noncommutative probability spaces (see \cite{JZ13,JZ14,SM14}) comes from the fact that in the non-tracial case, the Golden-Thompson inequality (see e.g.~\cite{Bhatia}) does not hold any longer, and its application is replaced by the use of operator Jensen's inequality.
		\end{remark}	
		
		The above proof can be simply extended to take into account the non-homogeneous case. We state the result in this case in the following corollary:
		\begin{corollary}\label{moderatenh}
			Let $\{\rho_n\}_{n\in\NN}$ and $\{\sigma_n\}_{n\in\NN}$ two sequences of states on $\rho_n,\sigma_n\in\cD(\cH^{\otimes n})$. If $\{\rho_n\}_{n\in\NN}$ and $\{\sigma_n\}_{n\in\NN}$ satisfy the non-homogeneous upper factorization property (\ref{upfac}), with $$0\le \log R<{(4-\e^c)(\e^c-1)^2 \liminf_nV_n(\{\tilde{\rho}_n\}\|\{\tilde{\sigma}_n\})/(6nc^2)},$$ and associated auxiliary sequences $\{\tilde{\rho}_n\}_{n\in\NN}$ and $\{\tilde{\sigma}_n\}_{n\in\NN}$ such that
				\begin{align*}
					c:= \sup_{n\in\NN} \|\log{\Delta_{\tilde{\sigma}_n|\tilde{\rho}_n}-D(\tilde{\rho}_n\|\tilde{\sigma}_n)}\|_\infty<\log 4 ,
				\end{align*}
				then for all $n$ large enough:
			\begin{align*}
				\frac{1}{n}D_H^{\eps_n}(\rho_n\|\sigma_n)\ge \frac{1}{n} D_n(\{\tilde{\rho}_n\}\|\{\tilde{\sigma}_n\})-\sqrt{\frac{{2V_n(\{\tilde{\rho}_n\}\|\{\tilde{\sigma}_n\})}}{{n}}}\left(\frac{3}{4-e^c}(\log R+a_n^2)\right)^{1/2},
			\end{align*}	
			where 
			\begin{align*}
				&D_n(\{\tilde{\rho}_n\}\|\{\tilde{\sigma}_n\}):=\sum_{k=1}^n D(\tilde{\rho}_k\|\tilde{\sigma}_k), ~~~~~V_n(\{\tilde{\rho}_n\}\|\{\tilde{\sigma}_n\}):=\sum_{k=1}^n V(\tilde{\rho}_k\|\tilde{\sigma}_k).
				\end{align*}
			If $\{\rho_n\}_{n\in\NN}$ and $\{\sigma_n\}_{n\in\NN}$ satisfy the non-homogeneous lower factorization property (\ref{lfp}) with $R>1$, then, 
			\begin{align}\label{eq27}
				\frac{1}{n}	D_H^{\eps_n}(\rho_n\|\sigma_n)\le \frac{1}{n}D_n(\{\tilde{\rho}_n\}\|\{\tilde{\sigma}_n\})-\sqrt{\frac{{2V_n(\{\tilde{\rho}_n\}\|\{\tilde{\sigma}_n\})}}{{n}}}a_n+\circ(a_n).
			\end{align}
		\end{corollary}

		\begin{remark}
			Theorem 1 of \cite{CTT17} (stated as \Cref{theoremCTT17} above) can be proved using the framework of relative modular operators by following similar steps as in the proof of \Cref{moderate}. The proof would only differ from the one above in that in the case of sequences of uncorrelated states, quantities of the form $(\log \Delta_{\sigma_n|\rho_n}+D(\rho_n\|\sigma_n)~\id)$ together with $\Omega_{\rho_n}$ can be directly associated to sums of independent, centered random variables. Therefore \Cref{prop_1} and \Cref{boundsD} suffice to reduce our problem to the one of finding asymptotic upper and lower bounds on tail probabilities of classical martingales. The lower bound on the $\eps_n$-hypothesis testing relative entropy therefore arises from a direct application of Bennett's inequality, whereas the lower bound arises from a direct application of the Berry-Esseen theorem. 
			\end{remark}

\section{Application to classical-quantum channels with memory}\label{channelcoding}

In \cite{CTT17}, the moderate deviation analysis of binary quantum hypothesis testing of pairs of sequences of uncorrelated states served as a tool to find asymptotic rates for transmission of information over a memoryless classical-quantum (c-q) channel\footnote{A channel is said to be memoryless if there is no correlation in the noise acting on successive inputs to the channel.}, subject to a sequence of tolerated error probabilities $\{\eps_n\}_{n\in\NN}$ vanishing sub-exponentially, with $\eps_n:=\e^{-na_n^2}$, for any moderate sequence $\{a_n\}_{n\in\NN}$ of real numbers. Here, we extend their results to a class of c-q channels with memory, described below. Let $\cH$ be a finite-dimensional Hilbert space and $\mathcal{X}$ be a (possibly uncountable) set of letters, called an \textit{alphabet}. In what follows we assume that $\cX$ is finite. By a \textit{classical-quantum} channel, we mean a map \[\mathcal{W}: \mathcal{X}\to \mathcal{D}(\cH),\]
and we denote its image by $\im(\mathcal{W})$. 
\\\\
Suppose that Alice (the sender) wants to communicate with Bob (the receiver) using the channel $\mathcal{W}$. To do this, they agree on a finite number of possible messages, labelled by the index set $\mathcal{M}:=\{1,...,M\}$. To send the message labelled by $k\in\mathcal{M}$, Alice encodes her message into a codeword $\phi(k)\equiv x_k\in\mathcal{X}$. The c-q channel $\mathcal{W}$ maps the codeword $x_k$ to a quantum state $\mathcal{W}(x_k)\in\cD(\cH)$, which Bob receives. To decode Alice's message, Bob performs a measurement, given by a POVM $T=\{T_1,...,T_M\}$ on $\cH$, where, for $i=1,2,...,M$, if the outcome corresponding to $T_i$ is obtained, he infers that the $i^{\text{th}}$ message was sent. If the message $k$ is sent, the probability of obtaining the outcome $l$ is given by
\begin{align*}
	\mathbb{P}(l|k)=\tr (\mathcal{W}(x_k)T_l),
	\end{align*} 
	and the average success probability of the encoding-decoding process is hence given by
	\begin{align*}
		\mathbb{P}(\operatorname{success}|\mathcal{W},T)=\frac{1}{M}\sum_{k=1}^M \tr(\mathcal{W}(x_k)T_k).
		\end{align*}
		For any fixed $\eps\in (0,1)$, the \textit{one-shot $\eps$-error capacity} of $\mathcal{W}$ is defined as follows:
		\begin{align*}
			C(\mathcal{W},\eps):=\log M^*(\mathcal{W},\eps),
			\end{align*}
			where,
		\begin{align}\label{Mstar}
			M^*(\mathcal{W},\eps):=\max~\{M\in \NN|~\exists \text{ POVM } T\equiv\{T_1,...,T_M\}:~ \mathbb{P}(\text{sucess}|\mathcal{W},T)\ge 1-\eps\}.
			\end{align}
		Capacities of c-q channels were originally evaluated in the asymptotic limit in which the channel is assumed to be available for arbitrary many uses. In the case of $n$ successive uses of a memoryless channel, the set of messages $\{1,...,M_n\}$ is encoded into $n$ letters, each belonging to a common alphabet $\mathcal{X}$. The encoding map is written as follows:
		\begin{align*}
			\phi_n:k\in\{1,...,M_n\}\mapsto \phi_n(k)=(x_{k,1},~...,~x_{k,n})\in\mathcal{X}^n.
			\end{align*}
		Then each letter $x\in\mathcal{X}$ is mapped to a state $\mathcal{W}(x)$. Then $n$ successive uses of the channel $\mathcal{W}$ map the sequence $(x_{k,1},~...,~x_{k,n})$ to 
		\begin{align*}
			\mathcal{W}^{\otimes n}(x_{k,1},~...,~x_{k,n})\equiv  \mathcal{W}(x_{k,1})\otimes~...~\otimes \mathcal{W}(x_{k,n}).
		\end{align*}
A natural extension of this framework is obtained by dropping the assumption of independence of successive uses of a single channel $\mathcal{W}$, thus allowing the channel to have memory.\\\\ In this section we consider a particular class of channels with memory defined as follows: let $\mathcal{W}_n:\mathcal{X}^{n}\to \mathcal{D}(\cH^{\otimes n})$ be a c-q channel where, for each $(x_1,~...,~x_n)\in \mathcal{X}^{ n}$, $\mathcal{W}_n(x_1,...,x_n)$ satisfies the non-homogeneous upper-factorization property
\begin{align}\label{ufpc}
	\mathcal{W}_n(x_1,...,x_n)\le R~\mathcal{W}_{n-1}(x_1,...,x_{n-1})\otimes \mathcal{W}_1(x_n).
\end{align}	
for some positive number $R$. We call this property of the maps $\mathcal{W}_n$ \textit{channel upper factorization property}.
%One can for example think of a message first encoded into some word of length $n$ whose letters belong to the alphabet $\mathcal{X}$, and then implemented into non-homogeneous finitely-correlated state on a spin chain of size $n$. 
Similarly, the assumption of independence of uses of a channel can be relaxed to incorporate another class of c-q channels with memory, whose outputs satisfy the non-homogeneous lower-factorization property:
\begin{align}\label{lfpc}
		R^{-1}\mathcal{W}_{n-1}(x_1,...,x_{n-1})\otimes \mathcal{W}_1(x_n)\le\mathcal{W}_n(x_1,...,x_n),
	\end{align}	
	and we call this property of the maps $\mathcal{W}_n$ \textit{channel lower factorization property}. \\\\
	We obtain bounds on the capacity of the above mentioned channels (i) for the case of finite blocklength (i.e.~finite $n$), as well as (ii) in the moderate deviation regime.	Our result (ii) extends the analysis of \cite{CTT17}, where asymptotic rates were found in the case of memoryless
c-q channels, to these new classes of channels with memory. As in \cite{CTT17}, the proofs of our results rely on bounds on the one-shot capacity of c-q channels obtained by Wang and Renner \cite{WR12} (stated as
\Cref{WR12} below) as well as the ones derived in Proposition 5 of \cite{TT15}. Here we make use of the notations of \cite{MO14}:
For every c-q channel $\mathcal{W}:\mathcal{X}\to \mathcal{D}(\cH)$, the following map
\begin{align*}
	\mathbb{W}:\mathcal{X}\to\cD(\cH_{\cX}\otimes\cH), ~~~~~~~\mathbb{W}(x):=|x\rangle\langle x|\otimes \mathcal{W}(x)
	\end{align*}
	is called the \textit{lifted channel} of $\mathcal{W}$, where $\cH_\cX$ is an auxiliary Hilbert space, and $\{|x\rangle:~x\in\cX\}$ is an orthonormal basis in it. The map $\mathcal{W}$ admits a natural linear extension to the set of probability mass functions $p_\cX$ on $\cX$ given by:
	\begin{align*}
		\mathcal{W}(p_\cX):=\sum_{x\in\cX}p_\cX(x) \mathcal{W}(x).
	\end{align*}	
	This extension can also be used at the level of the lifted channels as follows:
	\begin{align*}
		\mathbb{W}(p_\cX):=\sum_{x\in\cX}p_\cX(x)|x\rangle\langle x|\otimes \mathcal{W}(x).
	\end{align*}	

\begin{theorem}[\cite{WR12} Theorem 1]\label{WR12}
The $\eps$-error one-shot capacity of a c-q channel $\mathcal{W}:\mathcal{X}\to \cD(\cH)$ satisfies:
\begin{align*}
	C(\mathcal{W},\eps)\ge \sup_{p_\cX}D^{\eps'}_H(\rho_\cX\|\sigma_\cX) -\log\frac{4\eps}{(\eps-\eps')},
\end{align*}	
for every $\eps'\in(0,\eps)$, where for any finitely supported probability mass function $p_\cX$ on $\cX$, 
\begin{align*}
	\rho_\cX:=\mathbb{W}(p_\cX),~~~~~~~~~\sigma_\cX:=\sum_{x\in\cX}p_\cX(x)|x\rangle\langle x|\otimes \mathcal{W}(p_\cX).
	\end{align*}
\end{theorem}

In our proofs, we also use the following results that one can for example find in \cite{TT15}. Given a c-q channel $\mathcal{W}$, the Holevo capacity of $\mathcal{W}$
is given by
		\begin{align*}
			\chi^*(\mathcal{W}):= \min_{\sigma\in\cD(\cH)}~\sup_{x\in\cX}~D(\mathcal{W}(x)\|\sigma).
		\end{align*}
		The minimization over $\sigma\in\cD(\cH)$ is achieved for a unique state called the \textit{divergence centre} and denoted by $\sigma^*(\mathcal{W})$. The Holevo capacity has the following alternative representation
			\begin{align*}
				\chi^*(\mathcal{W})\equiv \sup_{p_\cX}~\sum_{x\in\cX} p_\cX(x)D\left(\mathcal{W}(x)\|\mathcal{W}(p_\cX)\right),
			\end{align*}
			and we denote the set of probability mass functions for which the above supremum is achieved by $\Pi(\mathcal{W})$. Now there exists a probability mass function $p_\mathcal{X}\in\Pi(\mathcal{W})$ such that \cite{TT15}
			\begin{align}\label{lll} D(\rho_{p_\mathcal{X}}\|\sigma_{p_{\mathcal{X}}})=\chi^*(\mathcal{W}),~~~~~~~~V(\rho_{p_{\mathcal{X}}}\|\sigma_{p_{\mathcal{X}}})=V_{\text{min}}(\mathcal{W}),
				\end{align} for states
	\begin{align}\label{rhosigma}
		&\rho_{p_\mathcal{X}}=\mathbb{W}(p_{\mathcal{X}}),~~~~~~~~~\sigma_{p_\mathcal{X}}=\sum_{x\in\mathcal{X}}p_{\mathcal{X}}(x)|x\rangle\langle x|\otimes \mathcal{W}(p_\mathcal{X}),
\end{align}
where
\begin{align*}
	V_{\text{min}}(\mathcal{W}):=\inf_{q_\mathcal{X}\in \Pi(\mathcal{W})}\sum_{x\in\mathcal{X}}q_\mathcal{X}(x) V(\mathcal{W}(x)\|\sigma^*(\mathcal{W})),
	\end{align*}
where $V(\cdot\|\cdot)$ is defined in \Cref{V}. We are interested in the finite blocklength behavior of the one-shot $\eps$-error capacity of the sequence of channels $\mathcal{W}_n$ that is the dependence of  $C(\mathcal{W}_n,\eps)$ on $n$. More precisely, we obtain lower bounds on $C(\mathcal{W}_n,\eps)$ in the following cases: \\
(i) $\mathcal{W}_n:=\mathcal{W}^{\otimes n}$ (memoryless case), and \\
(ii) $\mathcal{W}_n$ satisfies the channel upper factorization property (\ref{ufpc}). \\\\
The following result is a consequence of Theorem \ref{uncorcase}, Corollary \ref{cor2} and Theorem \ref{WR12} (for (i) one could alternatively use \Cref{theorem6} to get stronger bounds, but we omit this analysis for simplicity):
\begin{proposition}\label{prop}
	(i) Let $\mathcal{W}:\cX\to \cD(\cH)$ be a memoryless c-q channel. Then, for any $\eps\in(0,1)$, any $p_\cX\in\Pi(\mathcal{W})$, and any $\eps'\in(0,\eps)$:
	\begin{align*}
		C(\mathcal{W}^{\otimes n},\eps)\ge n\chi^*(\mathcal{W})-\sqrt{2n\log\eps'^{-1}}c_{p_\cX}-\log\frac{4\eps}{\eps-\eps'},
	\end{align*}	
	where $c_{p_\cX}:=\|\log\Delta_{\sigma_{p_\cX}|\rho_{p_\cX}}+D(\rho_{p_\cX}\|\sigma_{p_\cX})\id \|_{\infty}$, and $\rho_{p_\cX}$ and $\sigma_{p_\cX}$ are defined as in \Cref{rhosigma}.\\\\
	(ii) Let $\{\mathcal{W}_n\}_{n\in\NN}$ be a family of c-q channels $\mathcal{W}_n:\cX^{\otimes n}\to \cD(\cH^{\otimes n})$ satisfying the channel upper factorization property \reff{ufpc} with parameter $R\ge 1$. Then for any $\eps\in(0,1)$, any $p_\cX\in\Pi(\mathcal{W}_1)$ and any $\eps'\in(0,\eps)$:
	\begin{align*}
		C(\mathcal{W}_n,\eps)\ge \left\{\begin{aligned}
&n\chi^*(\mathcal{W}_1)-c_{p_\cX}\sqrt{2n\log(R^n\eps'^{-1})}-\log\frac{4\eps}{\eps-\eps'}~~~~\text{ for } R^n\e^{-nc_{p_\cX}^2/2}\le \eps'<\eps,\\
&n\chi^*(\mathcal{W}_1)-\frac{nc_{p_\cX}^2}{2}-\log R^n\eps'^{-1}-\log\frac{4\eps}{\eps-\eps'}~~~~~~~\text{ else.}
\end{aligned}
\right.
\end{align*}
\end{proposition}
\begin{proof}
We first prove (i): by a direct application of \Cref{WR12}, for any $\eps'\in(0,\eps)$,
\begin{align}
	C(\mathcal{W}^{\otimes n},\eps)&\ge \sup_{p_{\cX^{ n}}}D_H^{\eps'}(\rho_{p_{\cX^n}}\|\sigma_{p_{\cX^n}})-\log\frac{4\eps}{\eps-\eps'}\nonumber\\
&\ge \sup_{q_\cX}	D_H^{\eps'}(\rho_{q_\cX^{\otimes n}}\|\sigma_{q_\cX^{\otimes n}})-\log\frac{4\eps}{\eps-\eps'}\nonumber\\
&\ge D_H^{\eps'}(\rho_{p_\cX^{\otimes n}}\|\sigma_{p_\cX^{\otimes n}})-\log\frac{4\eps}{\eps-\eps'},\label{eq35}
\end{align}
for any probability mass function ${p_\cX}\in\Pi(\mathcal{W})$, where $D_H^\eps$ was defined in \Cref{DH}. %where in the second inequality we restricted to product distributions on $\cX^{n}$, and the third inequality comes by choosing the specific distribution $p_\cX$ of \Cref{lll}. Here, $\beta_n(\eps)$ denotes the optimal type II error in the quantum state discrimination problem with null hypothesis
%\begin{align*}
%\rho_{\cY^n}=\mathbb{W}(\mathbb{P}_{\cY}^{\otimes n})&=\sum_{(y_1,...,y_n)\in\cY^n}\mathbb{P}_{\cY}(y_1)~...~\mathbb{P}_{\cY}(y_n)|y_1~...~y_n\rangle\langle y_1~...~y_n|\otimes \mathcal{W}(y_1)\otimes~...~\otimes \mathcal{W}(y_n)\\
%&=\left(\sum_{y\in\cY}\mathbb{P}_\cX(y)|y\rangle\langle y|\otimes \mathcal{W}(y)\right)^{\otimes n}\equiv\rho_{\cY}^{\otimes n},
%\end{align*}	
%and alternative hypothesis 
%\begin{align*}
%	\sigma_{\cY^n}&=\sum_{(y_1,...,y_n)\in\cY^n}\mathbb{P}_{\cY}(y_1)~...~\mathbb{P}_{\cY}(y_n)|y_1~...~y_n\rangle\langle y_1~...~y_n|\otimes \mathcal{W}^{\otimes n}(\mathbb{P}_\cY^{\otimes n})\\
%	&=\left(\sum_{y\in\cY}\mathbb{P}_\cY(y)|y\rangle\langle y|\otimes \mathcal{W}(\mathbb{P}_\cY)\right)^{\otimes n}\equiv\sigma_{\cY}^{\otimes n}.
%	\end{align*}
Therefore, the problem reduces to the one of finding an upper bound to the optimal error of type II for the two i.i.d. sequences of states $\rho_{p_\cX^{\otimes n}}\equiv \rho_{p_\cX}^{\otimes n}$ and $\sigma_{p_{\cX}^{\otimes n}}\equiv \sigma_{p_\cX}^{\otimes n}$. From \reff{finite}, we directly get
\begin{align*}
	D_H^{\eps'}(\rho_{p_\cX^{\otimes n}}\|\sigma_{p_\cX^{\otimes n}})&\ge nD(\rho_{p_\cX}\|\sigma_{p_\cX})-\sqrt{2n\log\eps'^{-1}}c_{p_\cX},
\end{align*}	
and (i) follows from the fact that $p_\cX\in\Pi(\mathcal{W})$, so that $D(\rho_{p_\cX}\|\sigma_{p_\cX})\equiv \chi^*(\mathcal{W})$. \\\\
Case (ii) can be proved in a very similar way by noticing that in the case when $\mathcal{W}_n$ satisfies the channel upper factorization property \reff{ufpc}, the following states satisfy the upper factorization property: 
\begin{align*}
&\rho_{p_\cX^{\otimes n}}:=\mathbb{W}_n(p_{\cX}^{\otimes n})=\sum_{(x_1,...,x_n)\in\cX^n}p_\cX(x_1)~...~p_\cX(x_n)|x_1~...~x_n\rangle\langle x_1~...~x_n|~\mathcal{W}_n(x_1,...,x_n)\\
&~~~~~~~~~~~~~~~~~~~~~~~~~~~~~~~~~~~~~~~~~~~~~~~~~~~~~~~~~~~~~~~~~~~~~~~~~~~~~~~~~~~~~~~~~~\le R~ \rho_{p_\cX^{\otimes n-1}}\otimes \rho_{p_\cX},
\end{align*}
and 
\begin{align*}
			&\sigma_{p_\cX^{\otimes n}}:=\sum_{(x_1,...,x_n)\in\cX^n}p_{\cX}(x_1)~...~p_{\cX}(x_n)|x_1~...~x_n\rangle\langle x_1~...~x_n|\otimes \mathcal{W}_n(p_\cX^{\otimes n})\le R~ \sigma_{p_\cX^{\otimes n-1}}\otimes \sigma_{p_\cX},
	\end{align*}
	where we took $p_\cX$ to be the distribution such that $\rho_{p_\cX}$ and $\sigma_{p_\cX}$ satisfy \Cref{lll} for $\mathcal{W}\equiv \mathcal{W}_1$. We then obtain the statement of the theorem using \reff{eq35}, \Cref{ahi}, and \reff{lll}.		
	\qed
	\end{proof}
	\medskip
The moderate deviation analysis of the sequences of channels with memory defined above is given by the following proposition.

\begin{proposition}\label{eq29}
	(i) Let $\{\mathcal{W}_n\}_{n\in\NN}$ be a family of c-q channels $\mathcal{W}_n:\cX^{\otimes n}\to \cD(\cH^{\otimes n})$ satisfying the channel upper factorization property \reff{ufpc} with parameter $R$ such that $$1\le R< \exp[{(4-\e^{c_{p_\cX}})(1-\e^{c_{p_\cX}})^2 V_{\min}(\mathcal{W}_1)/(6c_{p_\cX}^2)}],$$ where $c_{p_\cX}:=\|\log\Delta_{\sigma_{p_\cX}|\rho_{p_\cX}}+D(\rho_{p_\cX}\|\sigma_{p_\cX})\id \|_{\infty}<\log 4$, for some probability mass function $p_\cX$ satisfying \Cref{lll}, and $\rho_{p_\cX}$ and $\sigma_{p_\cX}$ are defined as in \Cref{rhosigma}. Moreover, let $\{a_n\}_{n\in\NN}$ be a moderate sequence and $\eps_n=\e^{-na_n^2}$. Then,
	\begin{align}\label{lower}
		C(\mathcal{W}_n,\eps_n)\ge n\chi^*(\mathcal{W}_1)-n\sqrt{2V_{{\operatorname{min}}}(\mathcal{W}_1)}\left(\frac{3}{4-\e^{c_{p_\cX}}}(\log R+a_n^2)\right)^{1/2}+\circ(na_n).
		\end{align}
		(ii) Let $\{\mathcal{W}_n\}_{n\in\NN}$ be a family of c-q channels $\mathcal{W}_n:\cX^{\otimes n}\to \cD(\cH^{\otimes n})$ satisfying the channel lower factorization property \reff{lfpc} with parameter $R>1$. Then:
		\begin{align}\label{upper}
			C(\mathcal{W}_n,\eps_n)\le n\chi^*(\mathcal{W}_1)-\sqrt{2V_{\operatorname{min}}(\mathcal{W}_1)}na_n+\circ(na_n)
		\end{align}
\end{proposition}	
The proof of \Cref{eq29} closely follows that of Propositions 12 and 18 of \cite{CTT17}. The extra ingredient here is the fact that, if $\mathcal{W}_n$ satisfies a channel factorization property, the states in $\overline{\im(\mathcal{W}_n)}$ satisfy a non-homogeneous factorization property, so that one can use \Cref{moderatenh}, together with the one-shot bounds on the classical-quantum capacity to get the result. We included a proof in \Cref{thmm} for sake of completeness.

\section{Summary and discussion}
In this paper we proved upper bounds on type II Stein- and Hoeffding errors, in the context of finite blocklength binary quantum hypothesis testing,
using the framework of martingale concentration inequalities. These inequalities constitute a powerful mathematical tool which has found important applications in various branches of mathematics. We prove that our bounds are tighter than those obtained by Audenaert, Mosonyi and Verstraete in \cite{AMV12}, for a wide range of threshold values of the type I error, which is 
of practical relevance. We then derived finite blocklength bounds, as well as moderate deviation results, for pairs of sequences of correlated states satisfying a certain factorization property.  We applied our results to find bounds on the capacity of an associated class of classical-quantum channels with memory, in the finite blocklength, and moderate deviation regimes, This extends the recent results of Chubb, Tan and Tomamichel \cite{CTT17}, and Cheng and Hsieh \cite{HC17} to the non-i.i.d.~setting.\\\\
We believe that such extensions can be of practical relevance, for the following reasons: In the usual framework of quantum hypothesis testing, the systems that are being tested are demolished by the measurement process. However, in a practical experiment, the experimentalist might want to get some information about the physical system being analyzed without disturbing its state so that it can be used for some other information theoretic task. Although finitely correlated states were originally introduced in \cite{FNW92} in the context of quantum states on spin chains, they also seem to provide the right setup for what we call \textit{quantum non-demolition hypothesis testing} (QNDHT): In this framework, the goal is to determine the state of a quantum system, given the knowledge that it is one of two specific states $\rho,\sigma\in\cB(\mathcal{K})$, for a given finite-dimensional Hilbert space $\mathcal{K}$, without demolishing the system. A seemingly natural way to proceed is to prepare two sequences of finitely correlated states $\bar{\rho}_n,\bar{\sigma}_n\in\cD(\cH^{\otimes n})$, where $\cH$ is the Hilbert space of another system that can be interpreted as a probe, as follows:
\begin{align}
&	\bar{\rho}_n:=\tr_{\mathcal{K}} (\id_{\cA}^{n-1}\otimes \cE_*)\circ\dots \circ (\id_\cA\otimes \cE_*)\circ \cE_* (\rho),\\
	&\bar{\sigma}_n:=\tr_{\mathcal{K}} (\id_{\cA}^{n-1}\otimes \cE_*)\circ\dots \circ (\id_\cA\otimes \cE_*)\circ \cE_* (\sigma),
	\end{align}
where $\cE_*$ is the adjoint of a CPU map $\cE:\cB(\cH)\otimes \cB(\mathcal{K})\to \cB(\mathcal{K})$, to be specified later, encoding the interaction of the original system with the probes, such that
\begin{align*}
	\tr_\cH(\cE_*(\rho))=\rho~~\text{ and }~~	\tr_\cH(\cE_*(\sigma))=\sigma.
	\end{align*}
	This last condition precisely means that the original system  (with Hilbert space $\mathcal{K}$) should remain intact no matter what local operation (measurement) is done on the probes (whose Hilbert space is $\cH^{\otimes n}$). Note, however, that from \Cref{moderate}, we infer that the optimal Stein exponent in the quantum hypothesis problem with null hypotheses $\{\bar{\rho}_n\}_{n\in\NN}$ and alternative hypotheses $\{\bar{\sigma}_n\}_{n\in\NN}$ is given by 
\begin{align*}
	D(\bar{\rho}_1\|\bar{\sigma}_1)\le D(\rho\|\sigma),
	\end{align*}
where the last inequality follows from the data processing inequality. Intuitively, this means that the optimal error of type II made by measuring the probes is asymptotically larger than the error that one would make by performing a direct measurement on $n$ copies of the original system. A similar explanation holds for any fixed $n$, as the $\eps$-hypothesis testing relative entropy $D_H^\eps$ also satisfies a data processing inequality. An example of a map $\cE_*$ implementing the conditions described above can be described as follows: suppose without loss of generality that the system with Hilbert space $\mathcal{K}$ is in the state $\rho$. Then, at each step, make it interact with a probe $\cH$, which is initially in the state $\omega\in\cD(\cH)$:
\begin{align*}
	\rho\mapsto U (\rho\otimes \omega)U^*,
	\end{align*}
where $U$ is a unitary operator on $\mathcal{K}\otimes \cH$. In order to optimize the Stein exponent, we then consider the following optimization problem:\\\\ maximize
\begin{align*}
	D(\bar{\rho}_1\|\bar{\sigma}_1),
	\end{align*}
	 subject to \[\bar{\rho}_1=\tr_\mathcal{K}(U(\rho\otimes \omega)U^*),  ~~~~\bar{\sigma}_1=\tr_\mathcal{K}(U(\sigma\otimes \omega)U^*),\]
	  the optimization being carried over all states $\omega$ on $\cH$ and unitaries $U$ over $\cH\otimes \mathcal{K}$ satisfying
\begin{align*}
	\rho=\tr_\cH (U(\rho\otimes \omega)U^*),~~~~	\sigma=\tr_\cH( U(\sigma\otimes \omega)U^*).
	\end{align*}
	Secondly, c-q channels satisfying the channel factorization properties could potentially lead to new ways of efficiently implementing quantum communication channels. Indeed, Kastoryano and Brandao \cite{KB16} recently showed that, under some technical assumptions, Gibbs states of spin systems can be efficiently prepared by means of a dissipative process. Moreover, as discussed in \Cref{corr}, Gibbs states of translation-invariant finite-range interactions on quantum spin chains were shown to satisfy both lower and upper homogeneous factorization properties for $R>1$ \cite{HMO07}. We conjecture that by lifting the assumption of translation-invariance, one should obtain Gibbs states which satisfy both upper and lower non-homogeneous factorization properties. If this is indeed the case, the result of \cite{KB16} would provide an efficient way of implementing a c-q channel satisfying both lower and upper channel factorization properties whose capacity would be comparable (at least to leading order) to the one of memoryless c-q- channels (cf. \Cref{channelcoding}). The advantage of such a physical implementation comes from the robustness of such a dissipative preparation, in comparison with the difficulty of ensuring states to remain uncorrelated over a long period of time.

\paragraph{Acknowledgements}
The authors would like to thank Eric Hanson, Milan Mosonyi and Yan Pautrat for useful discussions. C.R.~is also grateful to Federico Pasqualotto
for helpful exchanges.

\bibliography{biblio}

\begin{thebibliography}{10}

\bibitem{A14}
Y.~Altug and A.~B. Wagner.
\newblock Moderate deviations in channel coding.
\newblock {\em IEEE Transactions on Information Theory}, 60(8):4417--4426,
  2014.

\bibitem{Araki76}
H.~Araki.
\newblock Relative entropy of states of von {N}eumann algebras.
\newblock {\em Publ. Res. Inst. Math. Sci.}, 11(3):809--833, 1975/76.

\bibitem{Araki77}
H.~Araki.
\newblock Relative entropy for states of von {N}eumann algebras. {II}.
\newblock {\em Publ. Res. Inst. Math. Sci.}, 13(1):173--192, 1977/78.

\bibitem{Aetal07}
K.~M.~R. Audenaert, J.~Calsamiglia, R.~Mu\~noz Tapia, E.~Bagan, L.~Masanes,
  A.~Acin, and F.~Verstraete.
\newblock Discriminating states: The quantum {Chernoff} bound.
\newblock {\em Phys. Rev. Lett.}, 98:160501, 2007.

\bibitem{AMV12}
K.~M.~R. Audenaert, M.~Mosonyi, and F.~Verstraete.
\newblock Quantum state discrimination bounds for finite sample size.
\newblock {\em Journal of Mathematical Physics}, 53(12), 2012.

\bibitem{A67}
K.~Azuma.
\newblock Weighted sums of certain dependent random variables.
\newblock {\em Tohoku Math. J. (2)}, 19(3):357--367, 1967.

\bibitem{Bhatia}
R.~Bhatia.
\newblock {\em Matrix analysis}, volume 169 of {\em Graduate Texts in
  Mathematics}.
\newblock Springer-Verlag, New York, 1997.

\bibitem{BLM13}
S.~Boucheron, G.~Lugosi, and P.~Massart.
\newblock {\em Concentration inequalities: A nonasymptotic theory of
  independence}.
\newblock Oxford university press, 2013.

\bibitem{BR1}
O.~Bratteli and D.~W. Robinson.
\newblock {\em Operator algebras and quantum statistical mechanics. 1}.
\newblock Texts and Monographs in Physics. Springer-Verlag, New York, second
  edition, 1987.

\bibitem{HC17}
H.-C. Cheng and M.-H. Hsieh.
\newblock Moderate deviation analysis for classical-quantum channels and
  quantum hypothesis testing.
\newblock {\em arXiv preprint arXiv:1701.03195}, 2017.

\bibitem{CTT17}
C.~T. Chubb, V.~Y. Tan, and M.~Tomamichel.
\newblock Moderate deviation analysis for classical communication over quantum
  channels.
\newblock {\em arXiv preprint arXiv:1701.03114}, 2017.

\bibitem{DPR16}
N.~Datta, Y.~Pautrat, and C.~Rouz\'{e}.
\newblock Second-order asymptotics for quantum hypothesis testing in settings
  beyond i.i.d.~{-}~quantum lattice systems and more.
\newblock {\em Journal of Mathematical Physics}, 57(6), 2016.

\bibitem{D57}
C.~Davis.
\newblock A schwarz inequality for convex operator functions.
\newblock {\em Proceedings of the American Mathematical Society}, 8(1):42--44,
  1957.

\bibitem{D10}
R.~Durrett.
\newblock {\em Probability: theory and examples}.
\newblock Cambridge university press, 2010.

\bibitem{EW16}
D.~Elkouss and S.~Wehner.
\newblock {(Nearly) optimal P values for all Bell inequalities}.
\newblock {\em Npj Quantum Information}, 2:16026, 2016.

\bibitem{FNW92}
M.~Fannes, B.~Nachtergaele, and R.~F. Werner.
\newblock Finitely correlated states on quantum spin chains.
\newblock {\em Comm. Math. Phys.}, 144(3):443--490, 1992.

\bibitem{HP81}
F.~Hansen and G.~K. Pedersen.
\newblock Jensen's inequality for operators and {L}\"{o}wner 's theorem.
\newblock {\em Mathematische Annalen}, 258:229--242, 1981.

\bibitem{HM7}
M.~{Hayashi}.
\newblock Error exponent in asymmetric quantum hypothesis testing and its
  application to classical-quantum channel coding.
\newblock {\em Phys. Rev. A}, 76:062301, 2007.

\bibitem{HMO07}
F.~Hiai, M.~Mosonyi, and T.~Ogawa.
\newblock Large deviations and {C}hernoff bound for certain correlated states
  on a spin chain.
\newblock {\em Journal of Mathematical Physics}, 48(12):123301, 2007.

\bibitem{HMO08b}
F.~Hiai, M.~Mosonyi, and T.~Ogawa.
\newblock {Error exponents in hypothesis testing for correlated states on a
  spin chain}.
\newblock {\em Journal of Mathematical Physics}, 49(3):032112, 2008.

\bibitem{HP91}
F.~Hiai and D.~Petz.
\newblock The proper formula for relative entropy and its asymptotics in
  quantum probability.
\newblock {\em Comm. Math. Phys.}, 143(1):99--114, 1991.

\bibitem{H63}
W.~Hoeffding.
\newblock Probability inequalities for sums of bounded random variables.
\newblock {\em Journal of the American Statistical Association},
  58(301):13--30, 1963.

\bibitem{IK10}
R.~Impagliazzo and V.~Kabanets.
\newblock Constructive proofs of concentration bounds.
\newblock In {\em Approximation, randomization, and combinatorial
  optimization}, volume 6302 of {\em Lecture Notes in Comput. Sci.}, pages
  617--631. Springer, Berlin, 2010.

\bibitem{Jaksicetal}
V.~Jak{\v{s}}i{\'c}, Y.~Ogata, C.-A. Pillet, and R.~Seiringer.
\newblock Quantum hypothesis testing and non-equilibrium statistical mechanics.
\newblock {\em Rev. Math. Phys.}, 24(6):1230002, 67, 2012.

\bibitem{JZ13}
M.~Junge and Q.~Zeng.
\newblock Noncommutative {B}ennett and {R}osenthal inequalities.
\newblock {\em Ann. Probab.}, 41(6):4287--4316, 11 2013.

\bibitem{JZ14}
M.~Junge and Q.~Zeng.
\newblock Noncommutative martingale deviation and {P}oincar\'{e} type
  inequalities with applications.
\newblock {\em Probability Theory and Related Fields}, pages 1--59, 2014.

\bibitem{KB16}
M.~J. Kastoryano and F.~G. S.~L. Brand{\~a}o.
\newblock Quantum gibbs samplers: The commuting case.
\newblock {\em Communications in Mathematical Physics}, 344(3):915--957, 2016.

\bibitem{KS98}
M.~Kearns and L.~Saul.
\newblock Large deviation methods for approximate probabilistic inference.
\newblock In {\em Proceedings of the Fourteenth conference on Uncertainty in
  artificial intelligence}, pages 311--319. Morgan Kaufmann Publishers Inc.,
  1998.

\bibitem{L05}
M.~Ledoux.
\newblock {\em The concentration of measure phenomenon}.
\newblock Number~89. American Mathematical Soc., 2005.

\bibitem{L14}
K.~Li.
\newblock Second-order asymptotics for quantum hypothesis testing.
\newblock {\em Ann. Statist.}, 42(1):171--189, 2014.

\bibitem{MD89}
C.~McDiarmid.
\newblock On the method of bounded differences.
\newblock {\em Surveys in combinatorics}, 141(1):148--188, 1989.

\bibitem{MO14}
M.~Mosonyi and T.~Ogawa.
\newblock Strong converse exponent for classical-quantum channel coding.
\newblock {\em arXiv preprint arXiv:1409.3562}, 2014.

\bibitem{MO15}
M.~Mosonyi and T.~Ogawa.
\newblock Two approaches to obtain the strong converse exponent of quantum
  hypothesis testing for general sequences of quantum states.
\newblock {\em IEEE Transactions on Information Theory}, 61(12):6975--6994,
  2015.

\bibitem{N06}
H.~{Nagaoka}.
\newblock The converse part of the theorem for quantum {Hoeffding} bound.
\newblock {\em arXiv preprint arXiv:0611289}, 2006.

\bibitem{NussbaumSzkola}
M.~Nussbaum and A.~Szko{\l}a.
\newblock The {C}hernoff lower bound for symmetric quantum hypothesis testing.
\newblock {\em Ann. Statist.}, 37(2):1040--1057, 2009.

\bibitem{OH04}
T.~Ogawa and M.~Hayashi.
\newblock On error exponents in quantum hypothesis testing.
\newblock {\em IEEE Trans. Inform. Theory}, 50(6):1368--1372, 2004.

\bibitem{ON00}
T.~Ogawa and H.~Nagaoka.
\newblock Strong converse and {S}tein's lemma in quantum hypothesis testing.
\newblock {\em IEEE Trans. Inform. Theory}, 46(7):2428--2433, 2000.

\bibitem{OP}
M.~Ohya and D.~Petz.
\newblock {\em Quantum entropy and its use}.
\newblock Texts and Monographs in Physics. Springer-Verlag, Berlin, 1993.

\bibitem{Petz1985}
D.~Petz.
\newblock {Quasientropies for states of a von Neumann algebra}.
\newblock {\em Publications of the Research Institute for Mathematical
  Sciences}, 21(4):787--800, 1985.

\bibitem{Petz1986}
D.~Petz.
\newblock {Quasi-entropies for finite quantum systems}.
\newblock {\em Reports on Mathematical Physics}, 23(1):57--65, 1986.

\bibitem{PV10}
Y.~Polyanskiy and S.~Verd{\'u}.
\newblock Channel dispersion and moderate deviations limits for memoryless
  channels.
\newblock In {\em Communication, Control, and Computing (Allerton), 2010 48th
  Annual Allerton Conference on}, pages 1334--1339. IEEE, 2010.

\bibitem{RS13}
M.~Raginsky and I.~Sason.
\newblock Concentration of measure inequalities in information theory,
  communications, and coding.
\newblock {\em Foundations and Trends® in Communications and Information
  Theory}, 10(1-2):1--246, 2013.

\bibitem{RS1}
M.~Reed and B.~Simon.
\newblock {\em Methods of Modern Mathematical Physics. {I}}.
\newblock Academic Press, Inc. [Harcourt Brace Jovanovich, Publishers], New
  York, second edition, 1980.

\bibitem{SM14}
G.~Sadeghi and M.~S. Moslehian.
\newblock Noncommutative martingale concentration inequalities.
\newblock {\em Illinois J. Math.}, 58(2):561--575, 2014.

\bibitem{S12}
I.~Sason.
\newblock Moderate deviations analysis of binary hypothesis testing.
\newblock In {\em 2012 IEEE International Symposium on Information Theory
  Proceedings}, pages 821--825, 2012.

\bibitem{TT96}
M.~Talagrand.
\newblock A new look at independence.
\newblock {\em Ann. Probab.}, 24(1):1--34, 01 1996.

\bibitem{TH13}
M.~Tomamichel and M.~Hayashi.
\newblock A hierarchy of information quantities for finite block length
  analysis of quantum tasks.
\newblock {\em IEEE Trans. Inform. Theory}, 59(11):7693--7710, 2013.

\bibitem{TT15}
M.~Tomamichel and V.~Y.~F. Tan.
\newblock Second-order asymptotics for the classical capacity of image-additive
  quantum channels.
\newblock {\em Communications in Mathematical Physics}, 338(1):103--137, 2015.

\bibitem{WR12}
L.~Wang and R.~Renner.
\newblock One-shot classical-quantum capacity and hypothesis testing.
\newblock {\em Phys. Rev. Lett.}, 108:200501, 2012.

\end{thebibliography}

\appendix

\section{A reverse Markov inequality}\label{reversemarkov}
\begin{lemma}[Reverse Markov inequality]
	Let $X$ be a strictly positive random variable, such that $1/X$ is integrable. Then, for any $x>0$,
	\begin{align*}
		\mathbb{P}(X>x)\ge \mathbb{E}\left[\frac{X-x}{X}\right].
		\end{align*}
	\begin{proof}
		For any decreasing bounded positive function $u$ such that $u^{-1}$ is also bounded,
		\begin{align*}
			\mathbb{P}(X>x)&=1-\PP(X\le x)\\
			&=1-\PP(u(X)\ge u(x))\\
			&\ge 1-\frac{\mathbb{E}[u(X)]}{u(x)},
			\end{align*}
			where we used Markov's inequality in the last line. Taking $u(x)=(1+tx)^{-1}$ for any given $t>0$,
			\begin{align*}
				\PP(X>x)&\ge 1-\EE[(1+tX)^{-1}](1+tx)=\mathbb{E}\left[\frac{t(X-x)}{1+tX}\right]
				=t\mathbb{E}\left[\frac{X-x}{1+tX}\right]=\EE\left[\frac{X-x}{1/t+X}\right].
				\end{align*}
				The result follows by monotone convergence theorem, taking the limit $t\to\infty$.
				\qed
	\end{proof}
\end{lemma}
\section{Proof of \Cref{nhcsupp}}\label{nhcsupp1}

The following lemma, originally proved in \cite{HMO07}, plays the key role in the proof of the factorization property of
non-homogeneous finitely correlated states.
\begin{lemma}[see Lemma 4.3 of \cite{HMO07}]\label{lemma1}
 Let $\cB\subset \cB(\mathcal{K})$ be a finite-dimensional $C^*$-subalgebra of $\cB(\mathcal{K})$ for some finite-dimensional Hilbert space $\mathcal{K}$. Let $\rho$ be a faithful state on $\cB$ and $\Phi:\cB\to\cB$ be the completely positive unital map $b\mapsto \tr(\rho b)\mathbb{I}_\cB$. Then there exists a constant $R>1$ such that $R\Phi-\id_\cB$ is completely positive. 
 \end{lemma}
 
 \begin{proof}[\Cref{nhcsupp}]
 	Using the notations of \Cref{corr}, let $\{\rho_n\}_{n\in\NN}$ be a family of non-homogeneous finitely correlated states with generating triple $(\cB,\{\cE_n\}_{n\in\NN},\rho)$. By \Cref{lemma1}, there exists $R>1$ such that $R~ \id_\cA^{\otimes n}\otimes (\cE_{n*}\circ\Phi_*)-\id_\cA^{\otimes n}\otimes \mathcal{E}_{n*}$ is positive for any $n\in\NN$, where $\Phi_*:a\mapsto \rho~\tr(a)$. Hence
 	\begin{align*}
 		\rho_n:=\tr_\cB \tau_n&\le R \tr_\cB[\id^{\otimes n-1}_\cA\otimes(\cE_{n*}\circ\Phi_*)(\tau_{n-1})]\\
 		&= R \tr_\cB[(\id_\cA^{\otimes n-1}\otimes \cE_{n*})(\tr_\cB(\tau_{n-1})\otimes \rho)]\\
 		&=R ~\rho_{n-1}\otimes \tr_\cB(\cE_{n*}(\rho)) \\
 		&=R  ~\rho_{n-1}\otimes \tilde{\rho}_n,
 		\end{align*}
 	where $\tilde{\rho}_n$ is defined as $\tr_\cB \cE_{n*}(\rho)$.
 	\qed
 	\end{proof}

 	\section{Proof of \Cref{eq29}}\label{thmm}
 
 For classical-quantum channels satisfying a channel factorization property, the expressions found in \Cref{moderatenh} can be used to get asymptotic behaviors for the capacity in the moderate deviation regime. The proof of \reff{upper} is more technical and follows the idea of \cite{CTT17} (see also \cite{TT15}). In order to prove it, we introduce the following geometric quantities: following \cite{TT15}, for some subset of states $\mathcal{S}\subseteq \cD(\cH)$, the \textit{divergence radius} $\chi(\mathcal{S})$ and \textit{divergence centre} $\sigma^*(\mathcal{S})$ are defined as
 \begin{align*}
 	\chi(\mathcal{S}):=\inf_{\sigma\in\cD(\cH)}\sup_{\rho\in\mathcal{S}}D(\rho\|\sigma),~~~~~~~\sigma^*(\mathcal{S}):=\argmin_{\sigma\in\cD(\cH)}\sup_{\rho\in\mathcal{S}}D(\rho\|\sigma).
 \end{align*}
 Similarly, the $\eps$-\textit{hypothesis testing divergence radius} $\chi_H^\eps(\mathcal{S})$ is defined as
 \begin{align*}
 	\chi^\eps_H(\mathcal{S}):=\inf_{\sigma \in\cD(\cH)}\sup_{\rho\in\mathcal{S}} D_H^\eps(\rho\|\sigma),
 \end{align*}
 The $\eps$-hypothesis testing divergence radius provides an upper bound to the one-shot capacity:
 \begin{theorem}[see Proposition 5 of \cite{TT15}]
 	The $\eps$-error one shot capacity, for $0<\eps<1/2$, is upper bounded as follows:
 	\begin{align}\label{u}
 		C(\mathcal{W},\eps)\le \chi_H^{2\eps}(\overline{\im(\mathcal{W})})+\log\frac{2}{(1-2\eps)(1-\eps)}.
 	\end{align}
 \end{theorem}
 We will also need the following lemma from \cite{TT15}:
 \begin{lemma}[Lemma 18 of \cite{TT15}]\label{lemmabis}
 	For every $\delta\in(0,1)$, there exists a finite set $\mathcal{S}^\delta\subset \cD(\cH)$ of cardinality
 	\begin{align*}
 		|\mathcal{S}^\delta|\le \left(\frac{90 d}{\delta^2}\right)^{2d^2},
 	\end{align*}	
 	where $d\equiv \dim(\cH)$, such that for every $\rho\in\cD(\cH)$, there exists a state $\tau\in\mathcal{S}^\delta$ such that
 	\begin{align*}
 		D(\rho\|\tau)\le\delta,~~~~~~~~~~~~ \lambda_{\operatorname{min}}(\tau)\ge \frac{\delta}{25d^2},
 	\end{align*}
 	where $\lambda_{\min}(\tau)$ stands for the minimum eigenvalue of the state $\tau$.
 \end{lemma}
 
 	\begin{proof}[\Cref{eq29}]
 		
 		The proof of (i) follows very closely the one of \Cref{prop} (ii), the only difference being the application of \reff{inward} instead of \reff{eq17} to bound the optimal error of type II.\\\\
 		We now prove (ii): by \reff{u}, 
 		\begin{align}\label{R}
 			C(\mathcal{W}_n,\eps_n)\le \chi_H^{2\eps_n}(\overline{\im(\mathcal{W}_n)})+\log\frac{2}{(1-2\eps_n)(1-\eps_n)}
 		\end{align}
 		$\overline{\im(\mathcal{W}_n)}$ consists of states satisfying a non-homogeneous lower factorization property with parameter $R>1$. Let $\rho_n\in\overline{\im(\mathcal{W}_n)}$, with associated auxiliary sequence $\{\tilde{\rho}_k\}_{k=1}^n$. Define \[\bar{\rho}_n:=\frac{1}{n}\sum_{k=1}^n\tilde{\rho}_k.\] For $\gamma$ a constant to be chosen later, define 
 		\begin{align*}
 			&	H(\gamma):=\left\{n:~ \frac{1}{n}\sum_{i=1}^n D(\tilde{\rho}_i\| \bar{\rho}_n)\ge \chi^*(\mathcal{W}_1)-\gamma\right\}, \\
 			&	 L(\gamma):=\left\{ n:~ \frac{1}{n}\sum_{i=1}^n D(\tilde{\rho}_i\|\bar{\rho}_n)<\chi^*(\mathcal{W}_1)-\gamma\right\},
 		\end{align*}
 		such that $H(\gamma)$ and $L(\gamma)$ bipartition $\NN$ for all $\gamma$. Let us define
 		\begin{align}\label{eq37}
 			\sigma_n(\gamma):=\frac{1}{2}[\sigma^*(\overline{\im(\mathcal{W}_1)})]^{\otimes n}+\frac{1}{2|\mathcal{S}^{\gamma/4}|}\sum_{\tau\in\mathcal{S}^{\gamma/4}}\tau^{\otimes n},
 		\end{align}
 		where $\mathcal{S}^{\gamma/4}$ is defined through \Cref{lemmabis}, with $\delta\equiv \gamma/4$. 
 		The idea of the proof is then to bound the divergences with respect to $\sigma_n(\gamma)$ by those with respect to either $\sigma^*(\overline{\im(\mathcal{W}_1)})$ or any element of $\mathcal{S}^{\gamma/4}$, using the following inequality:
 		\begin{align}\label{split}
 			D_H^\eps(\rho\|\mu\sigma+(1-\mu)\sigma')\le 	D_H^\eps(\rho\|\sigma)-\log\mu.
 		\end{align}	
 		Let us first assume that $n\in L(\gamma)$, and define $\tau_n$ to be the closest element in $\mathcal{S}^{\gamma/4}$ to $\overline{\rho}_n$, so that $D(\bar{\rho}_n\|\tau_n)\le \gamma/4$. Using \Cref{eq37}  we extract the term involving $\tau_n^{\otimes n}$ to obtain
 		\begin{align*}
 			D_H^{\eps_n}(\rho_n\|\sigma_n(\gamma))&\le D_H^{\eps_n}(\rho_n\|\tau_n^{\otimes n})-\log \frac{1}{2|\mathcal{S}^{\gamma/4}|}.
 		\end{align*}
 		Hence, for $n$ large enough, 
 		\begin{align}\label{eq32}
 			\frac{1}{n}	D_H^{\eps_n}(\rho_n\|\sigma_n(\gamma))\le\frac{1}{n} D_H^{\eps_n}(\rho_n\|\tau_n^{\otimes n})+\gamma/4.
 		\end{align}
 		Now applying \reff{eq27} we get that, for $n$ large enough,
 		\begin{align}\label{eq31}
 			\frac{1}{n}D_H^{\eps_n}(\rho_n\|\tau_n^{\otimes n})\le \frac{1}{n} \sum_{i=1}^n D(\tilde{\rho}_i\|\tau_n)-\sqrt{\frac{2\sum_{i=1}^nV(\tilde{\rho}_i\|\tau_n)}{n}}a_n+\gamma/4.
 		\end{align}
 	Now the following holds: 
 		\begin{align}
 			\sum_{i=1}^n D(\tilde{\rho}_i\|\tau_n)&=\sum_{i=1}^n\tr\tilde{\rho}_i(\log \tilde{\rho}_i-\log\bar{\rho}_n)+\sum_{i=1}^n\tr\tilde{\rho_i}(\log\bar{\rho}_n-\log\tau_n)\nonumber\\
 			&=\sum_{i=1}^n D(\tilde{\rho}_i\|\bar{\rho}_n)+nD(\bar{\rho}_n\|\tau_n)\nonumber\\
 			&\le \sum_{i=1}^nD(\tilde{\rho}_i\|\bar{\rho}_n)+n\gamma/4,\label{eq30}
 		\end{align}
 		where the last inequality above comes from the fact that we picked $\tau_n$ specifically so that $D(\bar{\rho}_n\|\tau_n)\le \gamma/4$.
 		Hence, combining \reff{eq32}, \reff{eq31} and \reff{eq30}, we have shown that for $n$ large enough:
 		\begin{align*}
 			\frac{1}{n}D_H^{\eps_n}(\rho_n\|\sigma_n(\gamma))\le \frac{1}{n}\sum_{i=1}^nD(\tilde{\rho}_i\|\bar{\rho}_n)-\sqrt{\frac{2\sum_{i=1}^nV(\tilde{\rho}_i\|\tau_n)}{n}}a_n+3\gamma/4.
 		\end{align*}
 		Finally, since $n\in L(\gamma)$, 
 		\begin{align}\label{eq36}
 			\frac{1}{n}	D_H^{\eps_n}(\rho_n\|\sigma_n(\gamma))\le \chi^*(\mathcal{W}_1)-\sqrt{\frac{2\sum_{i=1}^nV(\tilde{\rho}_i\|\tau_n)}{n}}a_n-\gamma/4.
 		\end{align}
 		We now take care of the case when $n\in H(\gamma)$. Let $\eta>0$. We use \Cref{eq37} to extract the term involving $\sigma^*(\overline{\im(\mathcal{W}_1)})$:
 		\begin{align*}
 			D_H^{\eps_n}(\rho_n\|\sigma_n(\gamma))\le  D_H^{\eps_n}(\rho_n\|[\sigma^*(\overline{\im(\mathcal{W}_1)})]^{\otimes n})+\log {2}.
 		\end{align*}
 		This implies that for any $\eta>0$ there exists $N>0$ such that for any $n>N$,
 		\begin{align}\label{eq33}
 		\frac{1}{n}	D_H^{\eps_n}(\rho_n\|\sigma_n(\gamma))\le\frac{1}{n}  D_H^{\eps_n}(\rho_n\|[\sigma^*(\overline{\im(\mathcal{W}_1)})]^{\otimes n})+\eta a_n/3.
 		\end{align}
 		Assume first that $V_{\min}(\mathcal{W}_1)\le \eta^2/18$. Applying \reff{eq27} for large enough $n$, 
 		\begin{align*}
 			\frac{1}{n}	D_H^{\eps_n}(\rho_n\|\sigma_n(\gamma))&\le\frac{1}{n}	D_H^{\eps_n}(\rho_n\|[\sigma^*(\overline{\im(\mathcal{W}_1)})]^{\otimes n})+\eta a_n/3\\
 			&\le \frac{1}{n}\sum_{i=1}^nD(\tilde{\rho}_i\|\sigma^*(\overline{\im(\mathcal{W}_1)}))+2\eta a_n/3\\
 				&\le \frac{1}{n}\sum_{i=1}^nD(\tilde{\rho}_i\|\sigma^*(\overline{\im(\mathcal{W}_1)}))-\sqrt{2V_{\operatorname{min}}(\mathcal{W}_1)}a_n+\eta a_n\\
 			&\le 
 			\chi^*(\mathcal{W}_1)-\sqrt{2V_{\operatorname{min}}(\mathcal{W}_1)}a_n+\eta a_n,
 		\end{align*}
 		where the last inequality follows from the definition of $\sigma^*(\overline{\im(\mathcal{W}_1)})$, and the fact that $\tilde{\rho}_i\in\overline{\im(\mathcal{W}_1)}$.\\\\For $V_{\operatorname{min}}(\mathcal{W}_1)> \eta^2/18$, consider the following quantity:
 			\begin{align*}
 				&\tilde{V}_{\operatorname{min}}(\gamma)\nonumber\\
 				&:=\inf_{p_\cX}\left\{\sum_{x\in\cX} p_\cX(x) V\left(\mathcal{W}_1(x)\|\sigma^*(\overline{\im(\mathcal{W}_1)})\right)\Bigg|~\sum_{x\in\cX}p_\cX(x)D\left(\mathcal{W}_1(x)\|\mathcal{W}_1(p_\cX)\right)\ge \chi^*(\mathcal{W}_1)-\gamma\right\},
 				\end{align*}
 			where the infimum is taken over any probability mass function $p_\cX$ on $\cX$. By definition $\tilde{V}_{\operatorname{min}}(0)=V_{\operatorname{min}}(\mathcal{W}_1)$, and using Lemma 22 of \cite{TT15}, we know that $\lim_{\gamma\to 0^+}\tilde{V}_{\operatorname{min}}(\gamma)=V_{\operatorname{min}}(\mathcal{W}_1)$. Hence, for any $\eta>0$ there exists a positive constant $\gamma_0$ such that
 			\begin{align}\label{eq34}
 				\sqrt{2\tilde{V}_{\operatorname{min}}(\gamma_0)}\ge \sqrt{2V_{\operatorname{min}}(\mathcal{W}_1)}-\eta/3.
 				\end{align}
 				As $V_{\operatorname{min}}(\mathcal{W}_1)> \eta^2/18$, this implies that $\tilde{V}_{\operatorname{min}}(\gamma_0)>0$.\\\\
 				Next, define the empirical probability mass function $p_n(x):=\frac{1}{n}\sum_{i=1}^n\delta(\mathcal{W}_1(x)-\tilde{\rho}_i)$. For all $n\in H(\gamma_0)$,
 				\begin{align*}
 				\sum_{x\in\cX}p_n(x)D\left(\mathcal{W}_1(x)\Bigg\|\sum_{y\in\cX}p_n(y)\mathcal{W}_1(y)\right)=\frac{1}{n}\sum_{i=1}^n D(\tilde{\rho}_i\|\bar{\rho}_n)\ge \chi^*(\mathcal{W}_1)-\gamma_0,
 					\end{align*}
 					and so we can lower bound the average quantum information variance with respect to the divergence centre
 					\begin{align*}
 						\frac{1}{n}\sum_{i=1}^n V(\tilde{\rho}_i\|\sigma^*(\overline{\im(\mathcal{W}_1)}))=\sum_{x\in\cX} p_n(x)V(\mathcal{W}_1(x)\|\sigma^*(\overline{\im(\mathcal{W}_1)}))\ge \tilde{V}_{\operatorname{min}}(\gamma_0)>0.
 						\end{align*}
 						Using this lower bound, we can once again apply \reff{eq27} to \reff{eq33} so that for $n\in H(\gamma_0)$ large enough,
 						\begin{align*}
 							\frac{1}{n} D_H^{\eps_n}(\rho_n\|\sigma_n(\gamma_0))&\le \frac{1}{n}D_H^{\eps_n}(\rho_n\|[\sigma^*(\overline{\im(\mathcal{W}_1)})]^{\otimes n})+\eta a_n/3\\
 							&\le \frac{1}{n}\sum_{i=1}^n D(\tilde{\rho}_i\|\sigma^*(\overline{\im(\mathcal{W}_1)}))-\sqrt{\frac{2\sum_{i=1}^nV(\tilde{\rho}_i\|\sigma^*(\overline{\im(\mathcal{W}_1)}))}{n}}a_n+2\eta a_n/3\\
 							&\le \frac{1}{n}
 								\sum_{i=1}^n D(\tilde{\rho}_i\|\sigma^*(\overline{\im(\mathcal{W}_1)}))-		\sqrt{2\tilde{V}_{\operatorname{min}}(\gamma_0)}a_n+2\eta a_n/3\\
 								&\le \chi^*(\mathcal{W}_1)-\sqrt{2V_{\operatorname{min}}(\mathcal{W}_1)}a_n+\eta a_n,
 							\end{align*}
 						where we used \reff{eq34} in the last line.
 				We showed that for any $\eta>0$ there exists $ \gamma_0$ such that for $n$ large enough
 				\begin{align*}
 					\frac{1}{n}D_H^{\eps_n}(\rho_n\|\sigma_n(\gamma_0))\le\left\{\begin{aligned}
 						&\chi^*(\mathcal{W}_1)-\gamma_0/4~~~~~~~~~~~~~~~~~~~~~~~~~~n\in L(\gamma_0),\\
 						&\chi^*(\mathcal{W}_1)-\sqrt{2V_{\operatorname{min}}(\mathcal{W}_1)}a_n+\eta a_n~~~n\in H(\Gamma).
 						\end{aligned}
 						\right.
 						\end{align*}
 		This implies that for any $n$ large enough:
 		\begin{align*}
 			\frac{1}{n}D_H^{\eps_n}(\rho_n\|\sigma_n(\gamma_0))\le \chi^*(\mathcal{W}_1)-\sqrt{2V_{\operatorname{min}}(\mathcal{W}_1)}a_n+\eta a_n
 			\end{align*}
 		Using the definition of $\chi_H$ and substituting the above inequality on the right hand side of \reff{R}, for $n$ large enough, we get
 		\begin{align*}
 			C(\mathcal{W}_n,\eps_n)&\le \inf_{\sigma_n\in\cD(\cH^{\otimes n})}\sup_{\rho_n\in\overline{\im(\mathcal{W}_n)}}D_H^{2\eps_n}(\rho_n\|\sigma_n)+\log \frac{2}{(1-\eps_n)(1-2\eps_n)}\\
 			&\le n\chi^*(\mathcal{W}_1)-\sqrt{2V_{\operatorname{min}}(\mathcal{W}_1)}na_n+\eta n a_n+\log \frac{2}{(1-\eps_n)(1-2\eps_n)}
 			\end{align*}
 			Taking the limit $\eta\to0$, we end up with,
 			\begin{align*}
 				C(\mathcal{W}_n,\eps_n)\le n\chi^*(\mathcal{W}_1)-\sqrt{2V_{\operatorname{min}}(\mathcal{W}_1)}na_n+\circ(na_n)
 				\end{align*}
 			
 		\qed
 	\end{proof}
\end{document}